\documentclass{amsart}

\usepackage[utf8]{inputenc} % allow utf-8 input
\usepackage[T1]{fontenc}    % use 8-bit T1 fonts
\usepackage{hyperref}       % hyperlinks
\usepackage{url}            % simple URL typesetting
\usepackage{booktabs}       % professional-quality tables
\usepackage{amsfonts}       % blackboard math symbols
\usepackage{nicefrac}       % compact symbols for 1/2, etc.
\usepackage{microtype}      % microtypography
\usepackage{lipsum}
\usepackage{amsthm}
\usepackage{amssymb}
\usepackage{amscd}
\usepackage{amsmath}
\usepackage[pdftex]{color,graphicx}
\usepackage{setspace}
\usepackage{cancel}
\usepackage{xcolor}
\usepackage[margin=2.5cm]{geometry}
\usepackage{comment}
\newtheorem{lemma}{Lemma}
\newtheorem{theorem}{Theorem}[section]
\newtheorem{corollary}[theorem]{Corollary}
\theoremstyle{definition}
\newtheorem{definition}[theorem]{Definition}

\theoremstyle{remark}
\newtheorem{remark}[theorem]{Remark}

\numberwithin{equation}{section}

\begin{document}

\title{On the geometry of electrovacuum spaces in higher dimensions}

\author{Maria Andrade$^{1}$}
\address{$^{1}$Universidade Federal de Sergipe, Departamento de Matemática, São Cristóvão, SE, Brasil, 49100-000.}
\email{maria@mat.ufs.br}

\author{Benedito Leandro$^{2}$}
\address{$^{2,\,3}$Instituto de Matem\'atica e Estat\'istica, Universidade Federal de Goi\'as, Goi\^ania, Brasil, 74690-900.}
\email{$^{2}$bleandroneto@ufg.br;\,$^{3}$robsonlousa@discente.ufg.br.}

\author{R\'obson Lousa$^{3}$}
\address{Róbson Lousa was supported by PROPG-CAPES [Finance Code 001].}

%    General info
%\subjclass[2010]{Primary 53C20, 53C21, 53C25}

%\keywords{Einstein-Maxwell equations, Electrostatic System, harmonic Weyl curvature, black holes}

\begin{abstract}
A classical question in general relativity is about the classification of regular static black hole solutions of the static Einstein-Maxwell equations (or electrovacuum system). In this paper, we prove some classification results for an electrovacuum system such that the electric potential is a smooth function of the lapse function. In particular, we show that an $n$-dimensional locally conformally flat extremal electrovacuum space must be in the Majumdar-Papapetrou class. Also, we prove that any three or four dimensional extremal electrovacuum space must be locally conformally flat. Moreover, we prove that an $n$-dimensional subextremal electrovacuum space with fourth-order divergence free Weyl tensor and zero radial Weyl curvature such that the electric potential is in the Reissner-Nordstr\"om class is locally a warped product manifold with $(n-1)$-dimensional Einstein fibers. Finally, a three dimensional subextremal electrovacuum space with third-order divergence free Cotton tensor was also classified.
\end{abstract}

\keywords{Eletrovaccum system, classification, conformally flat, radial Weyl tensor, Cotton tensor.} \subjclass[2020]{83C22, 83C05, 53C21.}

\maketitle

%\begin{document}
%\maketitle

\section{Introduction and Main Results}

Static electrovacuum spacetimes model exterior regions of static configurations of electrically charged stars or black holes (see \cite{cederbaum2016,chrusciel2007,hartle72} and the references therein). Equations of motion for an $(n+1)$-dimensional reduced Einstein-Maxwell spacetime are given by
	\begin{eqnarray*}
	(\widehat{\textnormal{Ric}})_{ij}=2\left(F_{il}F^{l}_{j}-\frac{1}{2(n-1)}|F|^{2}\widehat{g}_{ij}\right);\quad 1\leq i,\,j\leq n+1,
	\end{eqnarray*}
	where $F$ represents the electromagnetic field and $\widehat{\textnormal{Ric}}$ is the Ricci tensor for the metric $\widehat{g}.$ 
	
	Our main ground is the static space-time $(\widehat{M}^{n+1},\hat{g})=M^{n}\times_{f}\mathbb{R}$ such that 
	\begin{eqnarray*}
	\widehat{g}(x,\,t)=g(x)-f^{2}(x)dt^{2};\quad x\in M,
	\end{eqnarray*}
	where $(M^{n},g)$ is an open, connected and oriented Riemannian manifold, and $f$ is a smooth warped function. Considering as electromagnetic field $$F=d\psi\wedge dt,$$ for some smooth function $\psi$ on $M$ from the warped formula (see \cite{cederbaum2016,chrusciel1999,kunduri2018} and the references therein). The well known electrostatic (or electrovacuum) system is described below. 

\begin{definition}\label{def1}
Let $(M^n,g)$ be an $n$-dimensional smooth Riemannian manifold  with $n\geq 3$ and let $f,\psi:M\rightarrow\mathbb{R}$ be smooth functions satisfying 
%\begin{equation}\label{eq027}
\begin{equation}
\left\{\begin{array}{rcll}
\label{s1}
f\textnormal{Ric}&=&\nabla^2f-\dfrac{2}{f}d\psi\otimes d\psi+\dfrac{2}{(n-1)f}|\nabla\psi|^2g,\\\\
\Delta f&=&2\left(\dfrac{n-2}{n-1}\right)\dfrac{|\nabla\psi|^2}{f},\\\\
\textnormal{div}\left(\dfrac{\nabla\psi}{f}\right)&=&0,
 \end{array}\right.
\end{equation}
where \textnormal{Ric}, $\nabla^2$, \textnormal{div} and $\Delta$ are the Ricci and Hessian tensors, the divergence and the Laplacian operator on the metric $g$, respectively. Furthermore, $f > 0$ on $M$. Moreover, when $M^n$ has boundary $\partial M,$ we assume in addition  that $f^{-1}(0)= \partial M$. We also refer $(M^n, g, f, \psi)$ as electrovacuum (or electrostatic) system (or space).  The smooth functions $f$, $\psi$ and the manifold $M^n$ are called lapse function, electric potential and spatial factor for the static Einstein-Maxwell spacetime, respectively.
\end{definition}

We first observe that taking the contraction of the first equation and combining it with the second equation in \eqref{s1}, we obtain that the scalar curvature which is denoted by $R$ is given by
\begin{equation}\label{rrr}
    f^2R=2|\nabla\psi|^2.
\end{equation}
Second, when $\psi$ is a constant function, then the electrostatic system reduce to the static vacuum Einstein equations, i.e.,
	\begin{eqnarray}\label{vacuum}
	f\textnormal{Ric}=\nabla^{2}f\quad\mbox{and}\quad\Delta f=0.
	\end{eqnarray}
 These equations characterize the static vacuum Einstein spacetime which was widely explored in the literature. Furthermore, the most important solution for this system is the Schwarzschild solution. This solution represents a static black hole with mass, but without electric charge or magnetic fields. Therefore, we can see that Definition \ref{def1} generalizes the system \eqref{vacuum} and we will consider the case where $\psi$ is a constant function as trivial.

In 1918, independently, G. Nordström, and H. Reissner  found a class of exact solutions to the Einstein equation for the gravitational field of a spherical charged mass (see \cite{synge} for a wide-ranging discussion about these solutions). The  {\it Reissner-Nordstr\"om} (RN) electrostatic spacetime is one of the most important solutions for the electrostatic system  and it can be thought as a model for a static black hole or a star with electric charge $q$ and mass $m$. The RN spacetime is called subextremal, extremal or superextremal depending if $m^{2}>q^{2}$, $m^{2}=q^{2}$ or $m^{2}<q^{2}$, respectively. For instance, we have the following RN solution given by the Riemannian manifold $M^{n}=\mathbb{S}^{n-1}\times (r^{+},\,+\infty)$ with metric tensor
	$$g=\frac{dr^{2}}{1-2mr^{2-n}+q^{2}r^{2(2-n)}}+r^{2}g_{\mathbb{S}^{n-1}},$$
	where $r$ represents the radius of the Reissner-Nordstr\"om black hole. Here, $m^{2}\geq q^{2}$ are constants, and $r^{+}>(m+\sqrt{m^{2}-q^{2}})^{1/(n-2)}$.  Moreover, the outer horizon for the Reissner-Nordstr\"om space-time is located at $(m+\sqrt{m^{2}-q^{2}})^{1/(n-2)}$, which corresponds to the zero set of the lapse function of the RN manifold. The static horizon is defined as the set where the lapse function for a static manifold is identically zero. This set is physically related with the event horizon, the boundary of a black hole. The RN space is locally conformally flat (see \cite{chrusciel1999} for instance). 
 
 It is well known that the lapse
function $f$ and the electric potential $\psi$ of an electrovacuum system asymptotic to Reissner–Nordstr\"om of total mass $m$ and charge $q$, with suitable
inner boundary, satisfies the functional
relationship (see \cite[Equation A.1]{cederbaum2016} and \cite[Lemma 3]{kunduri2018})
\begin{eqnarray}\label{RNtype}
f^{2}=1+2\frac{n-2}{n-1}\psi^{2}-2\frac{m}{q}\sqrt{2\frac{n-2}{n-1}}\psi.
 \end{eqnarray}
 
 Another important electrovacuum solution is the {\it Majumdar–Papapetrou} (see \cite{chrusciel1999,hartle72,Lucietti}), which is related to an extremal RN solution. The Majumdar-Papapetrou (MP) solution to Einstein–Maxwell theory represents the static equilibrium of an arbitrary number of charged black holes whose mutual electric repulsion exactly balances their gravitational attraction. A spacetime will be called a standard MP spacetime if the metric tensor is given by
\begin{eqnarray}\label{MPmetric}
 \hat{g}=-f^{2}dt^{2}+f^{-2/(n-2)}(dx_1^{2}+\ldots+dx_n^{2}),
\end{eqnarray}
in Cartesian coordinates $x=(x_1,\,\ldots,\,x_n)$ and $\widehat{M}^{n+1}=(\mathbb{R}^{n}\backslash\{{a_i}\}_{i=1}^I)\times\mathbb{R},$ for a finite set of points ${a_i}\in\mathbb{R}^{n},$ where
\begin{eqnarray}\label{harmonic function MP}
 \frac{1}{f({x})}=1+\displaystyle\sum_{i=1}^{I}\frac{m_i}{r_i^{n-2}};\quad r_i=|{x}-{a_i}|,
\end{eqnarray}
 for some positive constants $m_i$, and the electric potential $$\psi=f,$$ (see \cite[Equation 2.3]{hartle72}).

The classification problem of an electrovacuum spacetime can be stated as follows. Suppose that
\begin{eqnarray*}\label{energycond}
\forall\,i,\,j\quad q_iq_j \geq 0,
\end{eqnarray*}
where $q_i$ is the charge of the $i$-th connected degenerate component of the electric charged black hole. Then
the black hole is either a RN black hole, or a MP black hole. There are some important and recent results in the literature concerning the classification of electrovacuum space (see for instance \cite{chrusciel2007,kunduri2018,Lucietti} and their references).

The most common assumption on the analysis and classification of an electrovacuum space is to consider that such space is asymptotically flat (see \cite{cederbaum2016,chrusciel1999,kunduri2018,Lucietti}). It is well known that using the positive mass theorem we can conclude that the space is conformally flat. We can then use classical calculations to prove that the solution for the electrovacuum system is either MP or RN (we refer to the reader see the main steps in the proof of \cite[Theorem 3.6]{chrusciel1999}). Those asymptotic conditions guarantee information about the metric, lapse function and the electric potential at infinity. Even though this condition is restrictive in the topological sense it is physically reasonable in the study of isolated gravitational system. Usually, in differential geometry some condition over the curvature is more often. However, considering just a condition over the curvature on the classification of the electrovacuum space seems to be not enough, since we lose information about the electric potential and the lapse function.

We recall that an asymptotically flat $n$-dimensional static Einstein-Maxwell space is extremal (i.e., $m=|q|$) if, and only if, the magnetic field is zero and $f=1\pm \sqrt{2(n-2)/(n-1)}\psi$, admitting $f=0$ at $\partial M$ (see Lemma 1 in \cite{kunduri2018}). Also, in \cite[Lemma 3]{kunduri2018}, the authors proved how the some kinds of electrovacuum solutions combined with an equation relating $\psi$ and $f$ have implications on the non-existence of magnetic fields. It is worth to say that an extremal RN space-time contains an unique photon sphere, on which light can get trapped and it has the largest possible ratio of charged to mass (see \cite{cederbaum2016}). The theory of extremal black holes is very important in physics and has very interesting properties. For instance, extremal charge black holes may be quantum mechanically stable, which is consistent with the ideas of cosmic censorship (see \cite{horowitz1991}). There is also an important type of electrovacuum solution in supergravity theory (see \cite{Lucietti}). Moreover, there is evidence that this type of black hole is important to the understanding of the no hair theorem (see \cite{burko2021}). 

The RN and MP solutions for the electrovacuum system suggest that there exists a class of solutions where the electric potential is a smooth function of the lapse function, i.e., $\psi=\psi(f)$. Our first result proves that there is a certain rigidity in this class of solutions.  
\begin{theorem}\label{psi de f}
Let $(M^n,g,f,\psi)$, $n\geq3$, be a complete electrovacuum space such that $\psi=\psi(f)$. Then, the electric potential (locally) is either 
\begin{eqnarray}\label{ruim}
\frac{2(n-2)}{n-1}\psi(f)^2-\frac{4(n-2)}{n-1}\beta\psi(f)+\frac{2(n-2)}{n-1}\beta^2+\frac{n-1}{n-2}\sigma=f^2
\end{eqnarray}
or
\begin{eqnarray}\label{boa}
 \psi(f)=\beta \pm \sqrt{\frac{(n-1)}{2(n-2)}}f,
\end{eqnarray}
where $\sigma,\,\beta\in\mathbb{R}.$ Moreover, $\sigma=0$ if and only if $\psi(f)$ is an affine function of $f$.
\end{theorem}

\begin{remark}
It is worth to highlight that the completeness hypothesis over $(M,\,g)$ is just to ensure that the critical set $\{\nabla f=0\}$ is not dense on $M$. So here, completeness can be replaced by assuming that the critical set is not dense.

%It is important to highlight that the hypothesis that the Riemannian manifold $(M,\,g)$ is complete is just to ensure that the critical set $\{\nabla f=0\}$ is not dense on $M$. So, here for us when we say complete we are referring to the critical set not being dense.
\end{remark}
It is interesting to remark how the constants $\sigma$ and $\beta$ given by \eqref{ruim} are related with the mass $m$ and electric charge $q$ for a RN solution which satisfies \eqref{RNtype}. A straightforward computation shows us that
$$\beta^2=\frac{(n-1)}{2(n-2)}\frac{m^2}{q^2}\quad\mbox{and}\quad\sigma=\frac{(n-2)}{(n-1)}\frac{q^2-m^2}{q^2}.$$
So, we can say that a solution satisfying \eqref{ruim} is called subextremal, extremal or superextremal depending on if $\sigma<0$, $\sigma=0$ or $\sigma>0$, respectively.

The above theorem shows us that an electrovacuum system such that $\psi=\psi(f)$ has basically two possible solutions and these solutions are closely related with the RN and MP solutions, respectively. It is also important to highlight that with the conformal metric $\widetilde{g}=f^{2/(n-2)}g$ the inverse of the electric potential $\frac{1}{\psi(f)}$ given by \eqref{boa} is harmonic in the metric $\widetilde{g}$. Moreover, $(M^{n},\,\widetilde{g})$ is Ricci-flat (see Lemma \ref{conflemma}). Then, considering asymptotic conditions, by the positive mass theorem, $(M^{n},\,\widetilde{g})$ is isometric to the Euclidean space. Of course, in three  dimensional case this is a direct consequence of $(M^{n},\,\widetilde{g})$ to be a Ricci-flat. This fact is important for the classification of extremal electrovacuum solutions. As pointed out in \cite[Remark 1]{kunduri2018} and \cite{Lucietti} any suitably regular asymptotically flat black hole solution in the Majumdar-Papapetrou class must has a space isometric to Euclidean space (minus a point for each horizon) and a harmonic function of the form \eqref{harmonic function MP}. In this case the spacetime is a Majumdar–Papapetrou multi-centred black hole solution (see \cite{Lucietti}). We need to emphasize that we are not considering any asymptotic conditions, so the positive mass theorem is not necessarily valid here. 

In the next result we prove that an extremal eletrovacuum space under certain hypothesis necessarily must be in the Majumdar-Papapetrou class.

\begin{theorem}\label{fiber0007}
Let $(M^{n},\,g,\,f,\,\psi)$, $n\geq3$, be an extremal electrovacuum space satisfying \eqref{boa}. Then, the Schouten tensor for the metric $g$ is Codazzi. If $(M^{n},\,g)$ is locally conformally flat, then any extremal solution must be in the Majumdar–Papapetrou class, i.e.,  $(M^{n},\,f^{2/(n-2)}g)$ is locally isometric to $\mathbb{R}^{n}$. Moreover, any three or four dimensional complete extremal electrovacuum space $(M,\,g)$ must be locally conformally flat.
\end{theorem}
We observe that if the Schouten tensor is Codazzi then in dimension three $(M^3,g)$ is locally conformally flat metric, already in dimension $n>3,$ then the Equation \ref{cw} implies in harmonic Weyl curvature. Codazzi tensors in Riemannian manifolds are important by themselves (see \cite[Proposition 16.11]{besse}). In addition, if an extremal electrovacuum solution is locally conformally flat, then is possible to use classical calculations to prove that it is a MP solution (see \cite[Proposition 3.4]{chrusciel1999}). Moreover, the extremal case was recently considered  in \cite{Lucietti}, where the author proved that the only asymptotically flat spacetimes with a suitably regular event horizon, in a generalised Majumdar–Papapetrou class of solutions to higher-dimensional Einstein–Maxwell theory are the standard multi-black holes \eqref{MPmetric}.

As a consequence of Theorem \ref{fiber0007} we need to highlight the following result.

\begin{corollary}
Any five dimensional extremal electrovacuum spacetime must be in the Majumdar–Papapetrou class.
\end{corollary}

Now, it remains to consider the electrovacuum solutions in the RN class, i.e., such that the electric potential is given by \eqref{ruim}. Here, we are considering divergence conditions on Weyl ($W$) and Cotton ($C$) tensors for a static Einstein-Maxwell spacetime instead of the traditional asymptotic conditions. Divergence conditions on $W$ have been recently explored in several works (see \cite{catino, catino2,benedito,qing} and the references therein). When the divergence of the Weyl tensor is identically zero, i.e., $$\textnormal{div}W=0,$$ we say that the manifold has \textit{a harmonic Weyl curvature}. It is well know that if the scalar curvature is constant, then harmonic Weyl curvature implies in harmonic curvature. This condition is equivalent to zero Cotton tensor in dimension more than $3$ (see \ref{cw}).

In what follows, we will consider that a Riemannian manifold $(M^n , g)$ has zero radial Weyl curvature if
\begin{equation}\label{radial}
    W(\cdot,\cdot,\cdot,\nabla f)=0,
\end{equation} 
where $\nabla f$ is the gradient for a smooth function $f:M\rightarrow\mathbb{R}.$ This condition was used in \cite{catino} and  \cite{benedito} in the study of Einstein-type manifolds, see more details in the references therein.

Now, we are ready to announce our next classification result.

\begin{theorem}\label{fiber007-1}
Let $(M^n, g, f, \psi)$, $n\geq3$, be a complete subextremal electrovacuum space with harmonic Weyl curvature and zero radial Weyl curvature such that $\psi$ is in the Reissner-Nordstr\"om class, i.e., such that $\psi$ is given by \eqref{ruim}. Then, around any regular point of $f$, the manifold is locally a warped product with $(n-1)$-dimensional Einstein fibers. 
\end{theorem}

\begin{remark}
In the three dimensional case, it is important to notice that the Weyl tensor $W$ is identically zero. So, the zero radial Weyl curvature condition is trivial. Moreover, the harmonic Weyl curvature condition must be replaced by locally conformally flat metric, i.e., $C=0$. 
\end{remark}
\begin{remark}
The subextremal condition required in the above theorem is just to avoid any major technical problem and can be relaxed by considering that $$f=\pm\sqrt{\frac{(n-1)}{(n-2)}\sigma}$$ is not dense at $M$. We are also considering that $\{f=0\}$ is not dense.
\end{remark}

In this paper, we will provide several results about divergence-free conditions in an electrovacuum space such that $\psi=\psi(f)$. Our goal is to provide a classification for an electrovacuum space having a fourth-order divergence-free Weyl tensor, i.e., $\textnormal{div}^{4}W=0$ (the space being compact or not). In the three dimensional case the discussion reduces to consider the Cotton tensor free from divergence, i.e., $\textnormal{div}^{3}C=0$. We will show that this higher-order divergence conditions can be reduced to harmonic Weyl curvature condition (or locally conformally flat curvature in the three dimensional case), under some additional hypothesis. 

The idea is to prove that the higher-order divergence-free conditions can be reduced to harmonic Weyl curvature (or zero Cotton tensor for $n=3$) using an appropriate divergence formula combined with some cut-off function and then, by integration of such formula, concluding that the Cotton tensor is identically zero, which is a similar strategy used by \cite{brendle,catino,benedito,qing}.

Next, as a consequence of Theorem \ref{fiber007-1} (see also Corollary \ref{fiber0072}), we get the following result.

\begin{corollary}\label{fiber007-2}
Let $(M^n, g, f, \psi)$, $n>3$, be a complete subextremal electrovacuum space with fourth-order divergence free Weyl curvature and zero radial Weyl curvature such that the electric potential $\psi$ is in the Reissner-Nordstr\"om class (i.e., satisfying Equation \eqref{ruim}). Around any regular point of $f$, if $f$ is a proper function, then the manifold is locally a warped product with $(n-1)$-dimensional Einstein fibers. 
\end{corollary}

In the three dimensional case the computations follow closely the same strategy of the above result and also we provide some interesting results reducing the order of divergence for the Cotton tensor. Let us show the most general case below (see partial results in the three dimensional compact space in Section \ref{sectionprincipaldim3}). In this way we obtain the following result.

\begin{corollary}
\label{fiber007-3}
Let $(M^3, g, f, \psi)$ be a complete subextremal electrovacuum space with third-order divergence free Cotton tensor such that $\psi$ is in the Reissner-Nordstr\"om class. Around any regular point of $f$, if $f$ is a proper function, then the manifold is locally an Einstein manifold, i.e., $(M^3,\,g)$ is locally isometric to either $\mathbb{R}^{3}$ or $\mathbb{S}^{3}$.
\end{corollary}

The paper is organized as follows. Section \ref{sec:Background} introduces terminology used throughout this paper. In Section \ref{section}, we present some structural lemmas that will be used in the proof of the main results. Finally, in Section \ref{secprincipal} we prove the main results.

\section{Background}
\label{sec:Background}

In this section, we fix our notation and recall some basic facts and useful lemmas. In particular, we need to remember some special tensors in the study of curvature for a Riemannian manifold $(M^n,\,g),\ n\geq 3.$ The first one is the Weyl tensor $W$ which is defined by
\begin{eqnarray}\label{wt}
     W_{ijkl}&=&R_{ijkl}-\frac{1}{n-2}(R_{ik}g_{jl}-R_{il}g_{jk}+R_{jl}g_{ik}-R_{jk}g_{il})\\
     &&+\frac{R}{(n-1)(n-2)}(g_{ik}g_{jl}-g_{il}g_{jk})\nonumber,
\end{eqnarray}
where $R_{ijkl}$ denotes the Riemann curvature tensor.
The second one, is the Cotton tensor given by
  \begin{equation}\label{ct}
  C_{ijk}=\nabla_iR_{jk}-\nabla_jR_{ik}-\frac{1}{2(n-1)}(\nabla_iRg_{jk}-\nabla_jRg_{ik}).
\end{equation}

And finally, considering $n\geq4$, the Bach tensor is defined by
\begin{equation}\label{bt}
   B_{ij}=\frac{1}{n-3}\nabla^k\nabla^lW_{ikjl}+\frac{1}{n-2}R^{kl}W_{ikjl}.
   \end{equation}
   We observe that the Weyl tensor has the same symmetries of the curvature tensor, that is
\begin{equation*}
 W_{ikjl}=-W_{kijl},\  W_{ikjl}=-W_{iklj}\ \text{and}\  W_{ikjl}=W_{jlik}.   
\end{equation*}
Moreover, we note that the Bach, the Cotton and the Weyl tensors are totally trace-free in any two indices (see \cite{cao} for instance), i.e., 
$$g^{ij}C_{ijk}=g^{ik}C_{ijk}=g^{jk}C_{ijk}=0.$$

When the dimension of $M$ is $n=3$, then the Weyl tensor $W_{ijkl}$ vanishes identically and the Cotton tensor $C_{ijk}=0$ if and only if $(M^3, g_{ij})$ is locally conformally flat; this fact holds if and only if $W_{ijkl}=0$, considering dimension $n\geq4$. Thus, for $n\geq4$ we have some well known relations with these tensors and their derivatives (see
\cite{cao, catino, benedito}). Involving the Weyl and Cotton tensors a straightforward computation yields to
\begin{equation}\label{cw}
    C_{ijk}=-\frac{n-2}{n-3}\nabla^lW_{ijkl}.
\end{equation}
So, if the Cotton tensor vanishes, then the Weyl tensor is harmonic.

 Now, for $n\geq 3$ combining (\ref{bt}) and (\ref{cw}) we can rewritten the Bach tensor as
\begin{equation}\label{bcw}
    B_{ij}=-\frac{1}{n-2}\nabla^kC_{ikj}+\frac{1}{n-2}R^{kl}W_{ikjl}.
\end{equation}
In particular, (see \cite{cao}), in dimension $n=3$, since the Weyl tensor identically zero, we can conclude that
\begin{equation}\label{bach3}
    B_{ij}=\nabla^kC_{kij}.
\end{equation}
This equation leads us to the following fact:
\begin{equation*}
    \nabla^kC_{kij}=\nabla^kC_{kji}.
\end{equation*}
Is convenient to express the divergence for the Bach tensor, which is given by
\begin{equation}\label{bc}
    \nabla^jB_{ij}=\frac{n-4}{(n-2)^2}C_{ijk}R^{jk}.
\end{equation}

Moreover, it is easy to see that 
$$C_{ijk}=-C_{jik}$$
and 
\begin{equation}\label{soma}
    C_{ijk}+C_{jki}+C_{kij}=0.
\end{equation}
From the contracted second Bianchi identity and from commutation formulas for any Riemannian manifold we can infer that  
\begin{equation}\label{zero}
    \nabla^iC_{jki}=0.
\end{equation}
Moreover, remember that $$(n-2){C}_{ijk}={\nabla}_{i}{S}_{jk}-{\nabla}_{j}{S}_{ik},$$
where $S$ stands for the Schouten tensor of $g$, i.e., 
\begin{eqnarray}\label{schoutentensor}
 S_{ij}=\frac{1}{n-2}\left(R_{ij}-\frac{R}{2(n-1)}g_{ij}\right).
\end{eqnarray}

%%%%%%%%%%%%%%%%%%%%%%%%%%%%%%%%%%%%%%%%%%%%%%%%%%%%%%%%%%%%%%%%%%%%%%%%%%%%%%%%%%%%%%%%%%%%%%%%%%%%%%%%%%%%%%%%%%%%%%%%%%%%%%%%%%%%%%%%%%%%%%%%%%%%%%%%%%%%%%%%%%%%%%%%%%
\iffalse
The next result is an important theorem about the classification of locally conformally flat manifolds and it will be very important here.
\begin{theorem}[\cite{zhu}]\label{zhu}
If $(M^n,g)$ is a complete locally conformally flat Riemannian manifold with $\textnormal{Ric}(g)\geq0$, then the universal cover $\widetilde{M}$ of $M$ with the pulled-back metric is either conformally equivalent to $\mathbb{S}^n$, $\mathbb{R}^n$ or is isometric to $\mathbb{R}\times \mathbb{S}^{n-1}$. If $M$ itself is compact, then the universal cover is either conformally equivalent to $\mathbb{S}^n$ or isometric to $\mathbb{R}^n$, $\mathbb{R}\times \mathbb{S}^{n-1}$, where $\mathbb{S}^n$ and $\mathbb{S}^{n-1}$ are spheres of constant curvature.
\end{theorem}
\fi

%%%%%%%%%%%%%%%%%%%%%%%%%%%%%%%%%%%%%%%%%%%%%%%%%%%%%%%%%%%%%%%%%%%%%%%%%%%%%%%%%%%%%%%%%%%%%%%%5

%%%%%%%%%%%%%%%%%%%%%%%%%%%%%%%%%%%%%%%%%%%%%%%%%%%%%%%%%%%%%%%%%%%%%%%%%%%%%%%%%%%%%%%%%%%%%%%%%%%%%%%%%%%%%%%%%%%%%%%%%%%%%%%%%%%%%%%%%%%%%%%%%%%%%%%%%%%%%%%%%%%%%%%%%%%%%%%%%%%%%%%%%%%%%%%%%%

\section{Structural lemmas}\label{section} 
Next, motivated by ideas in \cite{brendle,catino, benedito, qing} we obtain some structural lemmas,  which are fundamental to proof our results. Note that in a local coordinates system, using (\ref{rrr}) we can rewritten  \eqref{s1} as

\begin{eqnarray}
    fR_{jk}&=&\nabla_j\nabla_kf-\frac{2}{f}\nabla_j\psi\nabla_k\psi+\frac{1}{n-1}fRg_{jk};\label{1}\\
    \Delta f &=& \frac{n-2}{n-1}fR =  2\left(\frac{n-2}{n-1}\right)\frac{|\nabla \psi|^2}{f}\label{2};\\
    0&=&\Delta\psi-\frac{1}{f}\langle\nabla f,\,\nabla\psi\rangle.\label{3}
\end{eqnarray}

\begin{lemma}\label{primeiro}
Let $(M^n,\,g,\,f,\,\psi)$, $n\geq3$, be an electrovacuum system. Then,
    \begin{eqnarray*}
 fC_{ijk}&=&W_{ijkl}\nabla^lf+\frac{1}{n-2}(R_{jl}\nabla^lfg_{ik}-R_{il}\nabla^lfg_{jk})\\
 &+&\frac{R}{(n-1)(n-2)}(\nabla_ifg_{jk}-\nabla_jfg_{ik})\\
&-&\frac{2}{f^2}[f(\nabla_j\psi\nabla_i\nabla_k\psi-\nabla_i\psi\nabla_j\nabla_k\psi)-\nabla_if\nabla_j\psi\nabla_k\psi+\nabla_jf\nabla_i\psi\nabla_k\psi]\\
&+&\frac{n-1}{n-2}(R_{ik}\nabla_jf-R_{jk}\nabla_if)+\frac{1}{(n-1)f}(\nabla_i|\nabla\psi|^2g_{jk}-\nabla_j|\nabla\psi|^2g_{ik}).
    \end{eqnarray*}
\end{lemma}
\begin{proof}
We take the derivative of (\ref{1}) to obtain
\begin{eqnarray}\label{derivativei}
    R_{jk}\nabla_if+f\nabla_iR_{jk}&=&-\frac{2}{f^2}\left[f(\nabla_i\nabla_j\psi\nabla_k\psi+\nabla_j\psi\nabla_i\nabla_k\psi)-\nabla_if\nabla_j\psi\nabla_k\psi\right]\\
    &+&\nabla_i\nabla_j\nabla_kf+\frac{1}{n-1}\left(\frac{f}{2}\nabla_iR+\frac{1}{f}\nabla_i|\nabla\psi|^2\right)g_{jk}\nonumber
\end{eqnarray}
and
\begin{eqnarray}\label{derivativej}
R_{ik}\nabla_jf+f\nabla_jR_{ik}&=&-\frac{2}{f^2}\left[f(\nabla_j\nabla_i\psi\nabla_k\psi+\nabla_i\psi\nabla_j\nabla_k\psi)-\nabla_jf\nabla_i\psi\nabla_k\psi\right]\\
&+&\nabla_j\nabla_i\nabla_kf+\frac{1}{n-1}\left(\frac{f}{2}\nabla_jR+\frac{1}{f}\nabla_j|\nabla\psi|^2\right)g_{ik}.\nonumber
\end{eqnarray}
Subtracting (\ref{derivativei}) from (\ref{derivativej}) and using that the Hessian operator is symmetric, we can deduce that
\begin{eqnarray*}
R_{jk}\nabla_if-R_{ik}\nabla_jf&+&f(\nabla_iR_{jk}-\nabla_jR_{ik})=\nabla_i\nabla_j\nabla_kf-\nabla_j\nabla_i\nabla_kf+\frac{f}{2(n-1)}(\nabla_iRg_{jk}-\nabla_jRg_{ik})\\
     &-&\frac{2}{f^2}[f(\nabla_j\psi\nabla_i\nabla_k\psi-\nabla_i\psi\nabla_j\nabla_k\psi)-\nabla_if\nabla_j\psi\nabla_k\psi+\nabla_jf\nabla_i\psi\nabla_k\psi]\\
     &+&\frac{1}{(n-1)f}(\nabla_i|\nabla\psi|^2g_{jk}-\nabla_j|\nabla\psi|^2g_{ik}).
\end{eqnarray*}

It is well known that in any Riemannian manifold we can relate the Riemannian curvature tensor with a smooth function by using the Ricci identity
\begin{equation}\label{4}
   \nabla_i\nabla_j\nabla_kf-\nabla_j\nabla_i\nabla_kf=R_{ijkl}\nabla^lf.
\end{equation}

Then, replacing the Ricci identity \eqref{4} and the Cotton tensor (\ref{ct}), we infer that
\begin{eqnarray*}\label{auuu}
    fC_{ijk}&=&R_{ijkl}\nabla^lf+\frac{1}{(n-1)f}(\nabla_i|\nabla\psi|^2g_{jk}-\nabla_j|\nabla\psi|^2g_{ik})-R_{jk}\nabla_if+R_{ik}\nabla_jf\nonumber\\
     &-&\frac{2}{f^2}[f(\nabla_j\psi\nabla_i\nabla_k\psi-\nabla_i\psi\nabla_j\nabla_k\psi)-\nabla_if\nabla_j\psi\nabla_k\psi+\nabla_jf\nabla_i\psi\nabla_k\psi].
\end{eqnarray*}
Now, using the Weyl formula (\ref{wt}), we have
\begin{eqnarray*}
fC_{ijk}&=&W_{ijkl}\nabla^lf+\frac{1}{n-2}(R_{jl}\nabla^jfg_{ik}-R_{il}\nabla^lfg_{jk})-\frac{R}{(n-1)(n-2)}(g_{ik}\nabla^jf-g_{jk}\nabla^if)\\
&-&\frac{2}{f^2}[f(\nabla_j\psi\nabla_i\nabla_k\psi-\nabla_i\psi\nabla_j\nabla_k\psi)-\nabla_if\nabla_j\psi\nabla_k\psi+\nabla_jf\nabla_i\psi\nabla_k\psi]\\
    &+&\frac{n-1}{n-2}(R_{ik}\nabla^jf-R_{jk}\nabla^if)+\frac{1}{(n-1)f}(\nabla_i|\nabla\psi|^2g_{jk}-\nabla_j|\nabla\psi|^2g_{ik}).
\end{eqnarray*}
So, the proof is finished.
\end{proof}

In the sequel, we define the covariant $3$-tensor $V_{ijk}$ by
\begin{eqnarray}\label{tt}
V_{ijk}&=&\frac{1}{n-2}(R_{jl}\nabla^lfg_{ik}-R_{il}\nabla^lfg_{jk})+\frac{R}{(n-1)(n-2)}(\nabla_ifg_{jk}-\nabla_jfg_{ik})\nonumber\\     &-&\frac{2}{f^2}[f(\nabla_j\psi\nabla_i\nabla_k\psi-\nabla_i\psi\nabla_j\nabla_k\psi)-\nabla_if\nabla_j\psi\nabla_k\psi+\nabla_jf\nabla_i\psi\nabla_k\psi]\\
&+&\frac{n-1}{n-2}(R_{ik}\nabla_jf-R_{jk}\nabla_if)+\frac{1}{(n-1)f}(\nabla_i|\nabla\psi|^2g_{jk}-\nabla_j|\nabla\psi|^2g_{ik})\nonumber.
\end{eqnarray}
The tensor $V_{ijk}$ was defined similarly to $D_{ijk}$ in
\cite{cao}.
%\begin{equation*}
    %\begin{aligned}
      %T_{ijk}=&\frac{n-1}{n-2}(R_{ik}\nabla_jf-R_{jk}\nabla_if)+\frac{1}{n-2}(R_{jl}\nabla^lfg_{ik}-R_{il}\nabla^lfg_{jk})\\
        %&+\frac{R}{n-2}(\nabla_ifg_{jk}-\nabla_jfg_{ik}),
   % \end{aligned}
%\end{equation*}

Note that from a straightforward computation, we observe that the tensor $V$ has the same symmetries of the Cotton tensor $C$, i.e., 
$$V_{ijk}=-V_{jik}\quad\textnormal{and}\quad V_{ijk}+V_{jki}+V_{kij}=0.$$
 This $3$-tensor has a fundamental importance in what follows. From Lemma \ref{primeiro}, we have  
\begin{equation}\label{ttt}
    fC_{ijk}=W_{ijkl}\nabla^lf+V_{ijk}.
\end{equation}

In particular, if we suppose that $\psi=\psi(f)$ in the Lemma \ref{primeiro}, we obtain the following result.\\
 
\begin{lemma}\label{primeiro1}
Let $(M^n,\,g,\,f,\,\psi)$, $n\geq3$, be an electrovacuum system such that $\psi=\psi(f)$. Then,
\begin{eqnarray}\label{fundamental}
V_{ijk}=P(R_{il}\nabla^lfg_{jk}-R_{jl}\nabla^lfg_{ik})+Q(R_{ik}\nabla_jf-R_{jk}\nabla_if)+U(\nabla_ifg_{jk}-\nabla_jfg_{ik}),
\end{eqnarray}
where $$P=\dfrac{-1}{n-2}+\dfrac{2\dot{\psi}(f)^2}{n-1},\,\quad  Q=\dfrac{n-1}{n-2}-2\dot{\psi}(f)^2$$ and
$$U=\dfrac{R}{n-1}\left[\dfrac{1}{(n-2)}-\dfrac{2\dot{\psi}(f)^2}{(n-1)}+\dfrac{f\ddot{\psi}(f)}{\dot{\psi}(f)}\right].$$
\end{lemma}

\begin{proof} In fact, since $\psi=\psi(f)$, using \eqref{1}, we obtain
\begin{eqnarray*}\label{auuu1}
\nabla_k\nabla_i\psi&=&\ddot{\psi}(f)\nabla_kf\nabla_if+\dot{\psi}(f)\nabla_k\nabla_if\nonumber\\
&=&\ddot{\psi}(f)\nabla_kf\nabla_if+f\dot{\psi}(f)R_{ki}+\frac{2}{f}\dot{\psi}(f)^3\nabla_kf\nabla_if-\frac{1}{n-1}f\dot{\psi}(f)Rg_{ki}.
\end{eqnarray*}
Replacing the above equation in \eqref{tt} we can rewrite the $3$-tensor $V$ as
\begin{eqnarray}\label{tt3'}
V_{ijk}&=&\frac{1}{n-2}(R_{jl}\nabla^lfg_{ik}-R_{il}\nabla^lfg_{jk})+\left[\frac{R}{(n-1)(n-2)}-\frac{2}{n-1}\dot{\psi}(f)^2R\right](\nabla_ifg_{jk}-\nabla_jfg_{ik})\nonumber\\
&+&\left[\frac{n-1}{n-2}-2\dot{\psi}(f)^2\right](R_{ik}\nabla_jf-R_{jk}\nabla_if)+\frac{1}{(n-1)f}(\nabla_i|\nabla\psi|^2g_{jk}-\nabla_j|\nabla\psi|^2g_{ik}).
\end{eqnarray}

%%%%%%%%%%%%%%%%%%%%%%%%%%%%%%%%%%%%%%%%%%%%%%%%%%%%%%%%%%%%%%%%%%%%%%%%%%%%%%%%%%%%%%%%%%%%%%%%%%%%%%%%%
Now, by taking the derivative of \eqref{rrr} and using \eqref{eqnablapsi} we deduce that
$$ 4\ddot{\psi}(f)\dot{\psi}(f)\nabla_if|\nabla f|^2+2\dot{\psi}(f)^2\nabla_i|\nabla f|^2=2fR\nabla_if+f^2\nabla_iR.$$
Combining \eqref{eqnablapsi} and \eqref{1}, we obtain
\begin{eqnarray*}
 4\ddot{\psi}(f)\dot{\psi}(f)\nabla_if|\nabla f|^2+4\dot{\psi}(f)^2\left(fR_{il}\nabla_lf+\frac{2}{f}\dot{\psi}(f)^2\nabla_if|\nabla f|^2-\frac{1}{n-1}fR\nabla_if\right)=2fR\nabla_if+f^2\nabla_iR,
\end{eqnarray*}
this implies that
\begin{eqnarray}\label{derR}
 f^2\nabla_iR&=& 4\ddot{\psi}(f)\dot{\psi}(f)\nabla_if|\nabla f|^2+4\dot{\psi}(f)^2\left(fR_{il}\nabla_lf+\frac{2}{f}\dot{\psi}(f)^2\nabla_if|\nabla f|^2-\frac{1}{n-1}fR\nabla_if\right)\nonumber\\
 &-&2fR\nabla_if\nonumber\\
 &=&2fR\left(\frac{f\ddot{\psi}(f)}{\dot{\psi}(f)}+\frac{2(n-2)}{n-1}\dot{\psi}(f)^2-1\right)\nabla_if+4f\dot{\psi}(f)^2R_{il}\nabla_lf.
\end{eqnarray}
Then using \eqref{rrr} and \eqref{derR}, we get
\begin{eqnarray*}
 \nabla_i|\nabla\psi|^2&=&fR\nabla_if+\frac{f^2}{2}\nabla_iR\\
 &=&fR\left(\frac{f\ddot{\psi}(f)}{\dot{\psi}(f)}+\frac{2(n-2)}{n-1}\dot{\psi}(f)^2\right)\nabla_if+2f\dot{\psi}(f)^2R_{il}\nabla_lf.
\end{eqnarray*}

Thus, replacing this equation in \eqref{tt3'} the result follows.
\end{proof}

%%%%%%%%%%%%%%%%%%%%%%%%%%%%%%%%%%%%%%%%%%%%%%%%%%%%%%%%%%%%%%%%%%%%%%%%%%%%%%%%%%%%%%%%%%%%%%%%%%%%%%%

On the other hand, by the right conformal change of the metric we get our next lemma.

\begin{lemma}\label{conflemma}
Let $(M^n,g,f,\psi)$, $n\geq3$, be an electrovacuum system such that $\psi=\psi(f)$ is given by \eqref{boa}. Then, the Cotton tensor satisfies
\begin{eqnarray}\label{ish}
 (n-2)^{2}fC_{ijk}=W_{ijkl}\nabla^{l}f.
	\end{eqnarray}
In particular, when $n=3$, then $(M^{3},\,g)$ is locally conformally flat, i.e, $C=0$.
\end{lemma}
\begin{proof}
 We consider the conformal change of the metric $$\widetilde{g}=f^{\frac{2}{n-2}}g.$$
From \cite[Appendix]{catino} the Cotton tensor for metric $\widetilde{g}$ is given by
\begin{eqnarray}\label{cottontilde}
	(n-2)\widetilde{C}_{ijk}= (n-2)C_{ijk}-\frac{1}{(n-2)f}W_{ijkl}\nabla^{l}f.
\end{eqnarray}
Moreover, for $\widetilde{g}$ (see \cite[page 58]{besse}), we obtain
\begin{eqnarray}\label{ricconform}
\widetilde{\textnormal{R}}\textnormal{ic} &=& \textnormal{Ric} - \frac{1}{f}\nabla^{2}f +\frac{(n-1)}{(n-2)f^{2}}df\otimes df - \frac{\Delta f}{(n-2)f}g\nonumber\\
&=& \textnormal{Ric} - \frac{1}{f}\nabla^{2}f +\frac{(n-1)}{(n-2)f^{2}}df\otimes df - \frac{R}{(n-1)}g,
\end{eqnarray}
where in the last equation we used \eqref{2}.

Considering $\psi=\psi(f)$, from \ref{s1}, we get
\begin{eqnarray}\label{confchange}
	\widetilde{\textnormal{R}}\textnormal{ic} &=&\frac{1}{f^{2}}\frac{(n-1)}{(n-2)}df\otimes df -\frac{2}{f^{2}}d\psi\otimes d\psi+\frac{1}{(n-2)f}\left[2\frac{(n-2)}{(n-1)}\frac{|\nabla\psi|^{2}}{f}- \Delta f\right]f^{\frac{-2}{n-2}}\widetilde{g}\nonumber\\
	&=&\frac{1}{f^{2}}\frac{(n-1)}{(n-2)}df\otimes df -\frac{2}{f^{2}}d\psi\otimes d\psi
	=\frac{1}{f^{2}}\left[\frac{(n-1)}{(n-2)}-2\dot{\psi}^{2}\right]df\otimes df.
\end{eqnarray}
Moreover, 
\begin{eqnarray*}\label{confchangeR}
	\widetilde{R}=\frac{1}{f^{2}}\left[\frac{(n-1)}{(n-2)}-2\dot{\psi}^{2}\right]|\widetilde{\nabla}f|^{2}.
\end{eqnarray*}

By hypothesis $\psi=\psi(f)$ satisfies \eqref{boa}, then
\begin{eqnarray}\label{affinedot}
2\dot{\psi}^{2}=\frac{(n-1)}{(n-2)}.
\end{eqnarray} 
Consequently, from  \eqref{confchange} and \eqref{affinedot}, we conclude that $(M^{n},\,\widetilde{g})$ is Ricci-flat. In this case, the Schouten tensor for $\widetilde{g}$ is given by
\begin{eqnarray*}
\widetilde{S}&=&\frac{1}{n-2}\left(\widetilde{\textnormal{R}}\textnormal{ic}-\frac{1}{2(n-1)}\widetilde{R}\tilde{g}\right)\nonumber\\
&=&\frac{\left[\frac{(n-1)}{(n-2)}-2\dot{\psi}^{2}\right]}{(n-2)f^{2}}\left(df\otimes df-\frac{|\widetilde{\nabla} f|^{2}}{2(n-1)}\tilde{g}\right)=0.\nonumber\\
\end{eqnarray*}
This shows that $\widetilde{S}$ is Codazzi, because
$\widetilde{S}=0$, i.e., $(\widetilde{\nabla}_{X}\widetilde{S})(Y)=(\widetilde{\nabla}_{Y}\widetilde{S})(X)$ for all $X,\,Y\in TM$. Therefore, the Cotton tensor for the metric $\widetilde{g}$ is identicaly zero. So, from \eqref{cottontilde} we have
\begin{eqnarray*}\label{cottonpsiafim}
(n-2)^{2}fC_{ijk}=W_{ijkl}\nabla^{l}f.
	\end{eqnarray*}
Thus, we conclude our proof.
\end{proof}

Now, our goal is to obtain an useful formula for the norm of the Cotton tensor involving the divergence of the tensor $V$. To prove this, we need to show several lemmas. 
\begin{lemma}\label{segundo}
Let $(M^n,g,f,\psi)$, $n\geq4,$ be an electrovacuum system. Then,
\begin{equation}\label{lemma2}
    (n-2)B_{ij}=-\nabla^k\left(\frac{V_{ikj}}{f}\right)+\frac{n-3}{n-2}\frac{C_{jki}\nabla^kf}{f}+\frac{1}{f^2}W_{ikjl}(\nabla^kf\nabla^lf-2\nabla^k\psi\nabla^l\psi).
\end{equation}
\end{lemma}
\begin{proof}

In fact, from (\ref{bcw}) and (\ref{ttt}), we can deduce that
\begin{eqnarray*}
 (n-2)B_{ij}&=&-\nabla^kC_{ikj}+R^{kl}W_{ikjl}\\
&=&-\nabla^k\left(\frac{V_{ikj}}{f}+\frac{W_{ikjl}\nabla^{l}f}{f}\right)+R^{kl}W_{ikjl}\\
 &=&-\nabla^k\left(\frac{V_{ikj}}{f}\right)-\frac{\nabla^kW_{ikjl}\nabla^lf}{f}+\frac{W_{ikjl}\nabla^kf\nabla^lf}{f^2}-\frac{W_{ikjl}\nabla^k\nabla^lf}{f}+R^{kl}W_{ikjl}.
\end{eqnarray*}
    
Now, using (\ref{1}), we obtain
\begin{eqnarray*} (n-2)B_{ij}&=&-\nabla^k\left(\frac{V_{ikj}}{f}\right)-\frac{\nabla^kW_{ikjl}\nabla^lf}{f}+\frac{W_{ikjl}\nabla^kf\nabla^lf}{f^2}\\
&-&\frac{W_{ikjl}}{f}\left(fR^{kl}+\frac{2}{f}\nabla^k\psi\nabla^l\psi-\frac{1}{n-1}fRg^{kl}\right)+R^{kl}W_{ikjl}.
\end{eqnarray*}
Since the Weyl tensor is trace-free, we have
\begin{equation*}
(n-2)B_{ij}=-\nabla^k\left(\frac{V_{ikj}}{f}\right)-\frac{\nabla^kW_{ikjl}\nabla^lf}{f}+\frac{1}{f^2}W_{ikjl}(\nabla^kf\nabla^lf-2\nabla^k\psi\nabla^l\psi).
\end{equation*}
From \eqref{cw}, we get the result.
\end{proof}
Proceeding, we can use the previous lemma to obtain the following result.

\begin{lemma}\label{terceiro}
Let $(M^n,\,g,\,f,\,\psi)$, $n\geq4$, be an electrovacuum system. Then, 
\begin{eqnarray}\label{terceirol}
 C_{jki}R^{ik}&=&(n-2)\nabla^i\nabla^k\left(\frac{V_{ikj}}{f}\right)-(n-2)\frac{1}{f}W_{ikjl}R^{il}\nabla^kf\\ &+&2(n-2)\frac{W_{ikjl}}{f^2}\nabla^k\psi\nabla^i\nabla^l\psi-2(n-2)\frac{W_{ikjl}}{f^3}\nabla^if\nabla^k\psi\nabla^l\psi\nonumber.
\end{eqnarray}
\end{lemma}

\begin{proof}
Taking the divergence of \eqref{lemma2} and using \eqref{zero}, we infer that
\begin{eqnarray}\label{eqlem3}
(n-2)\nabla^iB_{ij}&=&-\nabla^i\nabla^k\left(\frac{V_{ikj}}{f}\right)+\frac{n-3}{n-2}\frac{C_{jki}}{f^2}(f\nabla^i\nabla^kf-\nabla^kf\nabla^if)\nonumber\\
&+&\frac{1}{f^2}W_{ikjl}\left(\nabla^i\nabla^kf\nabla^lf+\nabla^kf\nabla^i\nabla^lf-2\nabla^i\nabla^k\psi\nabla^l\psi-2\nabla^k\psi\nabla^i\nabla^l\psi\right)\\
&+&\frac{1}{f^2}\nabla^iW_{ikjl}\left(\nabla^kf\nabla^lf-2\nabla^k\psi\nabla^l\psi\right)-\frac{2}{f^3}W_{ikjl}\left(\nabla^if\nabla^kf\nabla^lf-2\nabla^if\nabla^k\psi\nabla^l\psi\right)\nonumber.
\end{eqnarray}
Since the Hessian is symmetric, then renaming indices and using the symmetries of the Weyl tensor, we deduce 
\begin{equation}\label{simpw}
 2\nabla^i\nabla^k\psi W_{ikjl}=\nabla^i\nabla^k\psi W_{ikjl}+\nabla^k\nabla^i\psi W_{kijl}=\nabla^i\nabla^k\psi(W_{ikjl}+W_{kijl})=0.
\end{equation}
Analogously, we have the same expression for the lapse function $f$, i.e.,
$$\nabla^i\nabla^kf W_{ikjl}=0.$$
Combining \eqref {eqlem3} and \eqref{simpw}, we obtain
\begin{eqnarray*} (n-2)\nabla^iB_{ij}&=&-\nabla^i\nabla^k\left(\frac{V_{ikj}}{f}\right)+\frac{n-3}{n-2}\frac{C_{jki}}{f^2}(f\nabla^i\nabla^kf-\nabla^kf\nabla^if)\\
&+&\frac{4}{f^3}W_{ikjl}\nabla^if\nabla^k\psi\nabla^l\psi-\frac{1}{f^2}\nabla^iW_{jlki}\left(\nabla^kf\nabla^lf-2\nabla^k\psi\nabla^l\psi\right)\nonumber\\
&+&\frac{1}{f^2}W_{ikjl}\left(\nabla^kf\nabla^i\nabla^lf-2\nabla^k\psi\nabla^i\nabla^l\psi\right).
\end{eqnarray*}

Since the Cotton and Weyl tensor are trace-free, using the symmetries of the Weyl tensor, \eqref{cw} and \eqref{1} we get
\begin{eqnarray}\label{eqnablab} (n-2)\nabla^iB_{ij}&=&-\nabla^i\nabla^k\left(\frac{V_{ikj}}{f}\right)+\frac{n-3}{n-2}C_{jki}R^{ik}+\frac{1}{f}W_{ikjl}R^{il}\nabla^kf\nonumber\\
&-&\frac{2}{f^2}W_{ikjl}\nabla^k\psi\nabla^i\nabla^l\psi+\frac{2}{f^3}W_{ikjl}\nabla^if\nabla^k\psi\nabla^l\psi.
\end{eqnarray}
Now, we need to remember some facts. Firstly, $B_{ij}=B_{ji},$  $R^{ij}=R^{ji}$ and the Cotton tensor is skew-symmetric, then we have an analogous relation to \eqref{simpw}, i.e,
\begin{equation}\label{anula}
    C_{ikj}R^{ik}=0.
\end{equation}
Secondly, using \eqref{soma}, we infer $C_{jik}=C_{jki}+C_{kij},$ this implies that $C_{jik}R^{ik}=C_{jki}R^{ik}.$ Thus, from \eqref{bc} and using these observations after renamed the indices, we obtain
\begin{equation*}
    \nabla^iB_{ij}=\frac{n-4}{(n-2)^2}C_{jik}R^{ik}=\frac{n-4}{(n-2)^2}C_{jki}R^{ik}.
\end{equation*}
Finally, using the above equation in \eqref{eqnablab} the result holds.
\end{proof}

\begin{lemma}\label{quarto}
Let $(M^n,g,f,\psi)$, $n\geq4$, be an electrovacuum system. Then, 
\begin{eqnarray*} \frac{1}{2}|C|^2+R^{ik}\nabla^jC_{jki}&=&(n-2)\nabla^j\nabla^i\nabla^k\left(\frac{V_{ikj}}{f}\right)-(n-2)\nabla^j\left[\frac{1}{f}W_{ikjl}R^{il}\nabla^kf\right]\\
        &-&2(n-2)\nabla^j\left[\frac{W_{ikjl}}{f^3}\nabla^if\nabla^k\psi\nabla^l\psi\right]+2(n-2)\nabla^j\left[\frac{W_{ikjl}}{f^2}\nabla^k\psi\nabla^i\nabla^l\psi\right].
\end{eqnarray*}
\end{lemma}
\begin{proof}
Taking the divergence of \eqref{terceirol}, we have
\begin{eqnarray}\label{eqlema4}
        C_{jki}\nabla^jR^{ik}+R^{ik}\nabla^jC_{jki}&=&(n-2)\nabla^j\nabla^i\nabla^k\left(\frac{V_{ikj}}{f}\right)-(n-2)\nabla^j\left[\frac{1}{f}W_{ikjl}R^{il}\nabla^kf\right]\nonumber\\
        &-&2(n-2)\nabla^j\left[\frac{W_{ikjl}}{f^3}\nabla^if\nabla^k\psi\nabla^l\psi\right]+2(n-2)\nabla^j\left[\frac{W_{ikjl}}{f^2}\nabla^k\psi\nabla^i\nabla^l\psi\right].
\end{eqnarray}
Note that from the symmetries of the Cotton tensor and renaming indices, we get
\begin{equation}\label{rc}
    2C_{jki}\nabla^jR^{ik}=C_{jki}\nabla^jR^{ik}+C_{kji}\nabla^kR^{ij}=C_{jki}(\nabla^jR^{ik}-\nabla^kR^{ij}).
\end{equation}
Then, combining \eqref{eqlema4} and \eqref{rc}, we can infer that
\begin{eqnarray*} \frac{1}{2}C_{jki}(\nabla^jR^{ik}-\nabla^kR^{ij})+R^{ik}\nabla^jC_{jki}&=&(n-2)\nabla^j\nabla^i\nabla^k\left(\frac{V_{ikj}}{f}\right)-(n-2)\nabla^j\left[\frac{1}{f}W_{ikjl}R^{il}\nabla^kf\right]\\
&-&2(n-2)\nabla^j\left[\frac{W_{ikjl}}{f^3}\nabla^if\nabla^k\psi\nabla^l\psi\right]\nonumber\\ &+&2(n-2)\nabla^j\left[\frac{W_{ikjl}}{f^2}\nabla^k\psi\nabla^i\nabla^l\psi\right].
\end{eqnarray*}

From \eqref{ct} and using that the Cotton tensor is trace-free, we obtain the result.
\end{proof}

\section{Proof of the main results}\label{secprincipal}
In this section, we prove our main results. 

\subsection{Classification Results}

Now, we are ready to present the proofs of Theorem \ref{psi de f}, Theorem \ref{fiber0007} and Theorem \ref{fiber007-1} which are the main classification results in this present work.  They will be stated again here for
the sake of the reader’s convenience.
%We start this subsection with the proof of Theorem \ref{fiber0007} and Theorem \ref{fiber007-1} which are the main classification results in this present work.

We start with Theorem \ref{psi de f} which shows us how related the electric potential and the lapse function can be. This result was inspired by \cite{kunduri2018} and \cite{ruback1988}.  

\begin{theorem}[Theorem \ref{psi de f}]
Let $(M^n,g,f,\psi)$, $n\geq3$, be a complete electrovacuum space such that $\psi=\psi(f)$. Then, the electric potential (locally) is either 
\begin{eqnarray*}\label{ruim'}
\frac{2(n-2)}{n-1}\psi(f)^2-\frac{4(n-2)}{n-1}\beta\psi(f)+\frac{2(n-2)}{n-1}\beta^2+\frac{n-1}{n-2}\sigma=f^2
\end{eqnarray*}
or
\begin{eqnarray*}\label{boa'}
 \psi(f)=\beta \pm \sqrt{\frac{(n-1)}{2(n-2)}}f,
\end{eqnarray*}
where $\sigma,\,\beta\in\mathbb{R}.$ Moreover, $\sigma=0$ if and only if $\psi(f)$ is an affine function of $f$.
\end{theorem}

\begin{proof}
%[\bf Proof of Theorem \ref{psi de f}]
Since $\psi=\psi(f),$ we obtain
\begin{equation}\label{eqnablapsi}
     \nabla \psi=\dot{\psi}(f)\nabla f.
\end{equation}
Then, 
\begin{eqnarray*}
\nabla^{2}\psi&=&\ddot{\psi}(f)df\otimes df+\dot{\psi}(f)\nabla^{2}f.
\end{eqnarray*}
Now, contracting the above equation over $g^{-1}$, we infer that
\begin{eqnarray*}
\Delta\psi&=&\ddot{\psi}(f)|\nabla f|^{2}+\dot{\psi}(f)\Delta f.
\end{eqnarray*}

From \eqref{2} and \eqref{eqnablapsi}, we have
\begin{eqnarray*}
    \Delta f=\frac{2}{f}\left(\frac{n-2}{n-1}\right)\dot{\psi}(f)^2|\nabla f|^2.
\end{eqnarray*}
Combining the last equations with \eqref{3} and \eqref{eqnablapsi}, we get
\begin{eqnarray*}
   \ddot{\psi}(f)|\nabla f|^{2}+2\left(\frac{n-2}{n-1}\right)\frac{\dot{\psi}(f)^3|\nabla f|^2}{f}=\dot{\psi}(f)\frac{|\nabla f|^{2}}{f}.
\end{eqnarray*}
Notice that there is no open subset $\Omega$ of $M$ where $\{\nabla f=0\}$ is dense. In fact, if $f$ is constant in $\Omega$, since $M$ is complete, then we have that $f$ is analytic, which implies $f$ is constant everywhere. By a straightforward computation, we arrive at 
\begin{eqnarray*}
  \dot{h}+2\left(\frac{n-2}{n-1}\right)fh^{3}=0,
\end{eqnarray*}
where $$h=\dfrac{\dot{\psi}}{f}.$$ So, by solving this ODE, we get
\begin{eqnarray}\label{psiponto}
   \dot{\psi}(f) = \frac{\pm f}{\sqrt{2\frac{(n-2)}{(n-1)}f^{2}-2\sigma}};\quad \sigma\in\mathbb{R}.
\end{eqnarray}

By integration we obtain, either 
\begin{eqnarray*}
 \psi(f)=\beta \pm \frac{(n-1)}{2(n-2)}\sqrt{2\left(\frac{n-2}{n-1}\right)f^{2}-2\sigma};\quad \  \sigma\neq 0,\,\beta\in\mathbb{R},
\end{eqnarray*}
or
\begin{eqnarray*}
 \psi(f)=\beta \pm \sqrt{\frac{(n-1)}{2(n-2)}}f;\quad \sigma=0\ ,\beta\in\mathbb{R}.
\end{eqnarray*}
Moreover, from \eqref{psiponto} we have the following useful identity
\begin{eqnarray}\label{2psiponto}
2\dot{\psi}(f)^2=\frac{(n-1)f^2}{(n-2)f^2-(n-1)\sigma}.
\end{eqnarray}

To finish, we observe that if $\sigma=0$, then from the above equation, $\dot{\psi}(f)$ is a constant, this implies that $\psi(f)$ is an affine function.
\end{proof}

In the next result we prove that an extremal eletrovacuum space under certain hypothesis necessarily must be in the Majumdar-Papapetrou class.

\begin{theorem}[Theorem \ref{fiber0007}]
Let $(M^{n},\,g,\,f,\,\psi)$, $n\geq3$, be an extremal electrovacuum space satisfying \eqref{boa}. Then, the Schouten tensor for the metric $g$ is Codazzi. If $(M^{n},\,g)$ is locally conformally flat, then any extremal solution must be in the Majumdar–Papapetrou class, i.e.,  $(M^{n},\,f^{2/(n-2)}g)$ is locally isometric to $\mathbb{R}^{n}$. Moreover, any three or four dimensional complete extremal electrovacuum space $(M,\,g)$ must be locally conformally flat.
\end{theorem}

\begin{proof}
%[\bf Proof of Theorem \ref{fiber0007}]
The proof follows from the previous section. In fact, remember that when $\psi$ is an affine function of $f$, we have the equation \eqref{affinedot}. Then, from \eqref{fundamental} we conclude that $P=Q=U=0$, and so the tensor $V_{ijk}$ is identically zero. Thus, from \eqref{ttt} we obtain  $fC_{ijk}=W_{ijkl}\nabla^{l}f.$ Immediately, for $n=3$ the Cotton tensor is identically zero which means that $(M^3,\, g,\,f,\,\psi)$ is locally conformally flat. 

Now, considering $n>3$, from the proof of Lemma \ref{conflemma} we obtain that the Ricci tensor, $\widetilde{\textnormal{R}}\textnormal{ic}$, for the conformal change of the metric $\widetilde{g}=f^{2/(n-2)}$ is identically zero, and so the Cotton tensor $\widetilde{C}_{ijk}$. At the same time, using \eqref{ish}, we can infer that $$(n-2)^{2}fC_{ijk}=W_{ijkl}\nabla^{l}f,$$ 
which combined with \eqref{ttt} gives us
$$[(n-2)^{2}-1]fC_{ijk}=0.$$
Consequently, the Schouten tensor \eqref{schoutentensor} is Codazzi, i.e., $(\nabla_{X}S)(Y)=(\nabla_{Y}S)(X)$ for all $X,\,Y\in TM.$ Furthermore, since $\widetilde{\textnormal{R}}\textnormal{ic}$ is identically zero, we conclude $(M^3,\, \widetilde{g})$ is isometric to $\mathbb{R}^{3}.$ 

Using again the conformal change of the metric $\widetilde{g}=f^{2/(n-2)}g$ (see \cite[page 58]{besse}), we have 
\begin{eqnarray}\label{sectionalconform}
\widetilde{R}_{ijkl}=f^{2/(n-2)}\Big[R_{ijkl} - (g_{ik}T_{jl}+g_{jl}T_{ik}-g_{il}T_{jk}-g_{jk}T_{il})\Big],
 \end{eqnarray}
 where
 \begin{eqnarray*}
  T_{ij}&=&\frac{1}{n-2}\left(\frac{1}{f}\nabla_i\nabla_jf-\frac{n-1}{(n-2)f^2}\nabla_if\nabla_jf+\frac{1}{2(n-2)f^2}|\nabla f|^2g_{ij}\right)\\
  &=&\frac{1}{(n-2)}\left(\frac{1}{f}\nabla_i\nabla_jf-\frac{(n-1)}{(n-2)f^{2}}\nabla_if\nabla_jf+\frac{R}{2(n-1)}g_{ij}\right).
 \end{eqnarray*}
In the last equality we used \eqref{2} and \eqref{affinedot}. 
%From Lemma \ref{conflemma}, we have $\widetilde{\textnormal{R}}\textnormal{ic}=0$.
Then,  from \eqref{ricconform},  we get 
\begin{eqnarray*}
 R_{ij} = \frac{1}{f}\nabla_i\nabla_jf -\frac{(n-1)}{(n-2)f^{2}}\nabla_if\nabla_jf+ \frac{R}{(n-1)}g_{ij}.
\end{eqnarray*}
Combining these two last identities we obtain
 \begin{eqnarray*}\label{tensorT}
  T_{ij}= \frac{1}{n-2}\left(R_{ij}-\frac{R}{2(n-1)}g_{ij}\right).
 \end{eqnarray*}
 Note that the tensor $T$ coincides with the Schouten tensor $S$ given by \eqref{schoutentensor}.
 If the Weyl tensor for $g$ is identically zero, then from \eqref{wt} we have
 \begin{eqnarray*}
g_{ik}T_{jl}+g_{jl}T_{ik}-g_{il}T_{jk}-g_{jk}T_{il}=R_{ijkl}.
 \end{eqnarray*}
 Therefore, replacing the above formula in \eqref{sectionalconform}, we can conclude that
  \begin{eqnarray*}
\widetilde{R}_{ijkl}=0.
 \end{eqnarray*}

We finish the proof considering the four dimensional case (see \cite[Lemma 4.3]{cao}). First, remember that in any open set of the level set $\Sigma=\{f=c\}$, where $c$ is any regular value for $f$, and using the local coordinates system
$$(x_1,\,x_2,\,x_3,\,x_4)=(f,\,\theta_2,\,\theta_3,\,\theta_4),$$ we can always express the metric ${g}$ in the form
\begin{eqnarray*}
	{g}_{ij} = \frac{1}{|\nabla f|^{2}}df^{2} + {g}_{ab}(f,\,\theta)d\theta_{a}d\theta_{b},
\end{eqnarray*}
where ${g}_{ab}(f, \theta)d\theta_{a}d\theta_{b}$ is the induced metric and $(\theta_{2},\,\theta_3,\,\theta_{4})$ is any local coordinate system on $\Sigma$ (see \cite[Remark 3.4]{cao}). In the following, we use $a,\,b,\,c$ to represent indices on the level sets which ranges from $2$ to $4$, while $i,\,j,\,k$ are used to represent indices on $M$ ranging from $1$ to $4$. Next, as it is well known that $\nu=\frac{-\nabla f}{|\nabla f|}$ is the normal vector field to $\Sigma$. Then is easy to see that
\begin{eqnarray*}
\nu=-|\nabla f|\partial_f\quad\mbox{or}\quad e_1:=\partial_f = \frac{\nabla f}{|\nabla f|^{2}}.
\end{eqnarray*}
  
Consider the referencial $\{e_1,\,e_{2},\,e_3,\,e_4\},$ where $e_{1}$ is normal and $e_a$ are tangent to $\Sigma$. Since the Schouten tensor is Codazzi, from \eqref{ttt} we have $W_{ijk1}=0$. Hence, we only need to show that 
$$W_{abcd}=0;\quad\forall  a,\,b,\,c,\,d\,\in\{2,\,3,\,4\}.$$
The Weyl tensor has all the symmetries of the curvature tensor and is trace-free in any two indices. Thus,
$$W_{2121}+W_{2222}+W_{2323}+W_{2424}=0,$$
this implies that 
$$W_{2323}=-W_{2424}.$$
Thus, from $$W_{2424}=-W_{3434}=W_{2323},$$
we conclude that $W_{2323}=0.$ Moreover,
$$W_{1314}+W_{2324}+W_{3334}+W_{4344}=0,$$
hence $W_{2324}=0.$
This shows that $W_{abcd} =0,$ unless $a,\,b,\, c,\, d$ are all distinct. But there are only three choices for the indices, since they range from $2$ to $4$. Then the Weyl tensor $W_{ijkl}$ is identically zero. Therefore, $(M^4,\,g)$ is locally conformally flat.

\end{proof}

Next, we prove Theorem  \ref{fiber007-1} about classification's result which was inspired by \cite{benedito}. 
%For completeness, we remember below the statement of this result.

\begin{theorem}[Theorem \ref{fiber007-1}]
Let $(M^n, g, f, \psi)$, $n\geq3$, be a complete subextremal electrovacuum space with harmonic Weyl curvature and zero radial Weyl curvature such that $\psi$ is in the Reissner-Nordstr\"om class, i.e., such that $\psi$ is given by \eqref{ruim}. Then, around any regular point of $f$, the manifold is locally a warped product with $(n-1)$-dimensional Einstein fibers. 
\end{theorem}

\begin{proof}
%[\bf Proof of Theorem \ref{fiber007-1}]
We consider an orthonormal frame $\{e_{1}, e_{2},\ldots,e_{n}\}$ diagonalizing the Ricci tensor $\textnormal{Ric}$ at a regular point $p\in\Sigma=f^{-1}(c)$, with associated eigenvalues $R_{kk}$, $k=1,\ldots, n,$ respectively. That is, $R_{ij}(p)=R_{ii}\delta_{ij}(p)$. From Lemma \ref{primeiro1}, we infer
	\begin{eqnarray}\label{22}
\nabla_{j}f[PR_{jj}+QR_{ii}-U]=0,\quad\forall i\neq j,
	\end{eqnarray}
	where $P,\,Q$ and $U$ are given by \eqref{fundamental}.
	Without lost of generalization, consider $\nabla_{i}f\neq0$ and $\nabla_{j}f=0$ for all $i\neq j$. Observe that $\textnormal{Ric}(\nabla f)=R_{ii}\nabla f$, i.e., $\nabla f$ is an eigenvector for $\textnormal{Ric}$. From (\ref{22}), we obtain that $\lambda=R_{ii}$ and $\mu=R_{jj}, \ j\neq i,$ have multiplicity $1$ and $n-1$, respectively. Moreover, if $\nabla_{i}f\neq0$ for at least two distinct directions, then using (\ref{22}) we have that $\mu=R_{11}=\ldots=R_{nn}$ and we also obtain that $\nabla f$ is an eigenvector for $\textnormal{Ric}$.

	Therefore, in any case we have that $\nabla f$ is an eigenvector for $\textnormal{Ric}$. From the above discussion we can take $\{e_{1}=\frac{-\nabla f}{|\nabla f|},e_{2},\ldots,e_{n}\}$ as an orthonormal frame for $\Sigma$ diagonalizing the Ricci tensor $\textnormal{Ric}$ for the metric $g$. 
	
Now, we note from \eqref{1} that
	\begin{eqnarray}\label{eqteo7}
	fR_{a1}=\frac{1}{2}\nabla_{a}|\nabla f|^{2}-\frac{2\dot{\psi}^{2}}{f}\nabla_af\nabla_bf+\frac{Rf}{(n-1)}\nabla_{a}f;\quad a\in\{2,\ldots,n\}.
	\end{eqnarray}
	Hence, equation \eqref{eqteo7} gives us $|\nabla f|$ is a constant in $\Sigma$. Thus, we can express the metric $g$ in the form
	 \begin{eqnarray*}
	 	g_{ij} = \frac{1}{|\nabla f|^{2}}df^{2} + g_{ab}(f,\theta)d\theta_{a}d\theta_{b},
	 \end{eqnarray*}
	 where $g_{ab}(f, \theta)d\theta_{a}d\theta_{b}$ is the induced metric and $(\theta_{2},\,\ldots,\,\theta_{n})$ is any local coordinate system on $\Sigma$. We can find a good overview of the level set structure in \cite{cao,benedito}.

Observe that there is no open subset $\Omega$ of $M^{n}$ where $\{\nabla f=0\}$ is dense. In fact, if $f$ is constant in $\Omega$, since $M^{n}$ is complete, we have that $f$ is analytic, which implies $f$ is constant everywhere. Thus, we consider $\Sigma$ a connected component of the level surface $f^{-1}(c)$ (possibly disconnected) where $c$ is any regular value of the function $f$. Suppose that $I$ is an open interval containing $c$ such that $f$ has no critical points in the open neighborhood $U_{I}=f^{-1}(I)$ of $\Sigma$. For sake of simplicity, let $U_{I}\subset M\backslash\{f=0\}$ be a connected component of $f^{-1}(I)$. Then, we can make a change of variables 
	 \begin{eqnarray*}
	 	r(x)=\int\frac{df}{|\nabla f|}
	 \end{eqnarray*}
	 such that the metric $g$ in $U_{I}$ can be expressed by
	 \begin{eqnarray*}
	 	g_{ij}=dr^{2}+g_{ab}(r,\theta)d\theta_{a}d\theta_{b}.
	 \end{eqnarray*}
	 
	 Let $\nabla r=\frac{\partial}{\partial r}$, then $|\nabla r|=1$ and $\nabla f=f'(r)\frac{\partial}{\partial r}$ on $U_{I}$. Note that $f^{\prime}(r)$ does not change
	 sign on $U_{I}$. Moreover, we have $\nabla_{\partial r}\partial r=0.$
	 
	 From \eqref{1} and the fact that $\nabla f$ is an eigenvector of $\textnormal{Ric}$, then the second fundamental formula on $\Sigma$ is given by
\begin{eqnarray*}\label{eq555}
h_{ab}&=& - \langle e_{1},\,\nabla_{a}e_{b}\rangle=\frac{\nabla_{a}\nabla_{b}f}{|\nabla f|}\nonumber\\
&=&\frac{1}{|\nabla f|}\left(f R_{ab}-\frac{Rf}{n-1}g_{ab}\right)=\frac{f}{|\nabla f|}\left(\mu-\frac{R}{n-1}\right)g_{ab}=\frac{H}{n-1}g_{ab},
\end{eqnarray*}
where $H=H(r)$, since $H$ is constant in $\Sigma$. In fact, contracting the Codazzi equation 
\begin{eqnarray*}
R_{1cab}=\nabla_{a}h_{bc}-\nabla_{b}h_{ac}
\end{eqnarray*}
over $c$ and $b$, it gives
\begin{eqnarray*}
R_{1a}=\nabla_{a}(H)-\frac{1}{n-1}\nabla_{a}(H)=\frac{n-2}{n-1}\nabla_{a}(H).
\end{eqnarray*}
On the other hand, since $R_{1a}=0,$ we conclude that $H$ is constant in $\Sigma$.

For what follows, we fix a local coordinates system
$$(x_{1},\, \ldots,\, x_{n}) = (r,\,\ldots,\, \theta_{n}) $$
in $U_{I}$, where $(\theta_{2},\ldots,\theta_{n})$ is any local coordinates system on the level surface $\Sigma_{c}$. Considering that $a, b, c,\cdots\in\{2, \ldots, n\}$, we have
\begin{eqnarray*}
h_{ab}=-g(\partial_{r},\, \nabla_{a}\partial_{b})=-g(\partial_{r}, \Gamma^{l}_{ab}\partial_{l})=-\Gamma^{1}_{ab}.
\end{eqnarray*}
Now, by definition
\begin{eqnarray*}
\Gamma^{1}_{ab}=\frac{1}{2}g^{11}\left(-\frac{\partial}{\partial r}g_{ab}\right)=-\frac{1}{2}\frac{\partial}{\partial r}g_{ab}.
\end{eqnarray*}
Then,
\begin{eqnarray*}
\frac{2}{n-1}H(r)g_{ab}=\frac{\partial}{\partial r}g_{ab}.
\end{eqnarray*}
Hence, we can infer that
\begin{eqnarray*}
g_{ab}(r,\theta)=\varphi(r)^{2}g_{ab}(r_{0},\theta),
\end{eqnarray*}
where $\varphi(r)=e^{\frac{1}{n-1}\left(\int^{r}_{r_{0}}H(s)ds\right)}$ and the level set $\{r=r_{0}\}$  corresponds to the connected component $\Sigma$ of $f^{-1}(c)$.

 Now, we can apply the warped product structure (see \cite{besse}). Hence, considering  $$(M^{n},\,g)=(I,\,dr^{2})\times_{\varphi}(N^{n-1},\,\bar{g})$$
 we deduce that
 \begin{eqnarray*}
 W_{1\,a\,1\,b}=\frac{1}{n-2}\bar{R}_{a\,b}-\frac{\bar{R}}{(n-2)(n-1)}g_{a\,b}.
 \end{eqnarray*} Finally, if $W(\cdot,\,\cdot,\,\cdot,\,\nabla f)=0$ we obtain that $N$ is an Einstein manifold. Remember that, $N$ is an Einstein manifold if and only if $M$ is also Einstein. 
\end{proof}

%%%%%%%%%%%%%%%%%%%%%%%%%%%%%%%%%%%%%%%%%%%%%%%%%%%%%%%%%%%%%%%%%%%%%%%%%%%%%%%%%%%%%%%%%%%%%%%%%%%%%%%%%%%%%%%%%%%%%%%%%%%%%%%%%%%%%%%%%%%%%%%%%%%%%%%%%%%%%%%%%%%%%%%%%%%%%%%%
%\textcolor{red}{@Maria I stopped here!}

\subsection{Fourth-order divergence free Weyl tensor}
%In this section, our aim is to prove some integral theorems in dimension $n\geq4,$  \textcolor{red}{} elevate the divergence order of freedom for the Weyl tensor for an electrovacuum space in the RN class. To that end, the proved lemmas in the previous section will perform a key role. In our results we are considering subextremal Riemannian manifolds satisfying the zero radial Weyl curvature. Indeed, the fact that electrovacuum space can not be extremal appears naturally in the first theorem of this section.  

In this subsection, our aim is to prove some integral theorems in dimension $n\geq4$ with fourth-order divergence-free Weyl tensor for an electrovacuum space in the RN class. To that end, we use the lemmas in the previous section. In our results we are considering subextremal Riemannian manifolds satisfying the zero radial Weyl curvature. Indeed, the fact that electrovacuum space can not be extremal appears naturally in the first theorem of this subsection.

\begin{theorem}\label{teo3}
Let $(M^n, g, f, \psi)$, $n \geq 4$, be an electrovacuum space satisfying \eqref{s1}, \eqref{ruim} and \eqref{radial}. For every $\phi:\mathbb{R}\rightarrow\mathbb{R}$, $C^2$
function with $\phi(f)$ having compact support $K\subseteq M$. Then,
\begin{eqnarray*}
\frac{1}{2(n-1)^2\sigma}\int_M|C|^2\phi(f)\left[(n-1)\sigma-(n-2)f^2\right]&=&-\frac{n-2}{n-3}\int_M\frac{\phi(f)}{f}\nabla^kf\nabla^i\nabla^j\nabla^lW_{jkil},
\end{eqnarray*}
where $\sigma$ is a non-null constant.

\end{theorem}
\begin{remark}\label{remarktheo}
It is important to point out that the choice of $\phi$ in the above theorem should be made in such way that terms like $\dfrac{\phi(f)}{f^{m}},$ where $m=1,\,2$ or $3$, will be integrable at $K$.
\end{remark}
\begin{proof}
From Lemma \ref{quarto}, we have
\begin{eqnarray*}
\frac{1}{2}|C|^2\phi(f)+\phi(f)R^{ik}\nabla^jC_{jki}&=&(n-2)\phi(f)\nabla^j\nabla^i\nabla^k\left(\frac{V_{ikj}}{f}\right)-(n-2)\phi(f)\nabla^j\left[\frac{1}{f}W_{ikjl}R^{il}\nabla^kf\right]\\
&-&2(n-2)\phi(f)\nabla^j\left[\frac{W_{ikjl}}{f^3}\nabla^if\nabla^k\psi\nabla^l\psi\right]+2(n-2)\phi(f)\nabla^j\left[\frac{W_{ikjl}}{f^2}\nabla^k\psi\nabla^i\nabla^l\psi\right].
\end{eqnarray*}

Now, an integration by parts  leads us to
\begin{eqnarray*}
\frac{1}{2}\int_M|C|^2\phi(f)+\int_M\phi(f)R^{ik}\nabla^jC_{jki}&=&-(n-2)\int_M\dot{\phi}(f)\nabla^jf\nabla^i\nabla^k\left(\frac{V_{ikj}}{f}\right)\\
        &+&(n-2)\int_M\dot{\phi}(f)\nabla^jf\left[\frac{1}{f}W_{ikjl}R^{il}\nabla^kf\right]\\
        &-&2(n-2)\int_M\dot{\phi}(f)\nabla^jf\left[\frac{W_{ikjl}}{f^2}\nabla^k\psi\nabla^i\nabla^l\psi)\right]\\
        &+&2(n-2)\int_M\dot{\phi}(f)\nabla^jf\left[\dfrac{W_{ikjl}}{f^3}\nabla^if\nabla^k\psi\nabla^l\psi\right].
\end{eqnarray*}
From Lemma \ref{terceiro}, we obtain 
\begin{equation*}
        \frac{1}{2}\int_M|C|^2\phi(f)+\int_M\phi(f)R^{ik}\nabla^jC_{jki}=-\int_M\dot{\phi}(f)\nabla^jfC_{jki}R^{ik}.
\end{equation*}
Using (\ref{1}), we deduce that
\begin{eqnarray*}
\frac{1}{2}\int_M|C|^2\phi(f)&+&\int_M\frac{\phi(f)}{f}\left(\nabla^i\nabla^kf-\frac{2}{f}\nabla^i\psi\nabla^k\psi+\frac{1}{n-1}fRg_{ik}\right)\nabla^jC_{jki}=\\
&=&-\int_M\frac{\dot{\phi}(f)}{f}\nabla^jf\left(\nabla^k\nabla^if-\frac{2}{f}\nabla^i\psi\nabla^k\psi+\frac{1}{n-1}fRg_{ik}\right)C_{jki}.
\end{eqnarray*}
Since the Cotton tensor is totally trace-free, we can infer that
\begin{eqnarray}\label{inicio}
\frac{1}{2}\int_M|C|^2\phi(f)&+&\int_M\frac{\phi(f)}{f}\nabla^i\nabla^kf\nabla^jC_{jki}-2\int_M\frac{\phi(f)}{f^2}\nabla^i\psi\nabla^k\psi\nabla^jC_{jki}\nonumber\\
&=&-\int_M\frac{\dot{\phi}(f)}{f}\nabla^jf\nabla^k\nabla^ifC_{jki}+2\int_M\frac{\dot{\phi}(f)}{f^2}\nabla^jf\nabla^i\psi\nabla^k\psi C_{jki}.
\end{eqnarray}

Analogously to \eqref{simpw}, we have the following equation
\begin{equation}\label{hessiana}
2\nabla^j\nabla^k\psi C_{jki}=\nabla^k\nabla^j\psi C_{jki}+\nabla^j\nabla^k\psi C_{kji}=\nabla^k\nabla^j\psi(C_{jki}+C_{kji})=0.
\end{equation} 
Then, using this relation, we get
\begin{eqnarray*} -2\int_M\frac{\phi(f)}{f^2}\nabla^i\psi\nabla^k\psi\nabla^jC_{jki}&=&2\int_M\left(\frac{\dot{\phi}(f)}{f^2}-\frac{2\phi(f)}{f^3}\right)\nabla^jf\nabla^i\psi\nabla^k\psi C_{jki}\\
&+&2\int_M\frac{\phi(f)}{f^2}\nabla^j\nabla^i\psi\nabla^k\psi C_{jki}.
\end{eqnarray*}
Replacing it in (\ref{inicio}), since the Cotton tensor is skew-symmetric, renaming indices we obtain
\begin{eqnarray*}
\frac{1}{2}\int_M|C|^2\phi(f)&+&\int_M\frac{\phi(f)}{f}\nabla^i\nabla^kf\nabla^jC_{jki}-4\int_M\frac{\phi(f)}{f^3}\nabla^jf\nabla^i\psi\nabla^k\psi C_{jki}\\
&+&2\int_M\frac{\phi(f)}{f^2}\nabla^j\nabla^i\psi\nabla^k\psi C_{jki}=
-\int_M\frac{\dot{\phi}(f)}{f}\nabla^jf\nabla^k\nabla^ifC_{jki}\\
&=&\int_M\frac{\dot{\phi}(f)}{f}\nabla^jf\nabla^if\nabla^kC_{jki}+\int_M\left(\frac{\ddot{\phi}(f)}{f}-\frac{\dot{\phi}(f)}{f^2}\right)C_{jki}\nabla^if\nabla^kf\nabla^jf\\
&+&\int_M\frac{\dot{\phi}(f)}{f}\nabla^k\nabla^jf\nabla^ifC_{jki}\\
&=&\int_M\frac{\dot{\phi}(f)}{f}\nabla^jf\nabla^if\nabla^kC_{jki}\\
&=&-\int_M\frac{\dot{\phi}(f)}{f}\nabla^if\nabla^kf\nabla^jC_{jki}=-\int_M\frac{\nabla^i\phi(f)}{f}\nabla^kf\nabla^jC_{jki}\\
&=&\int_M\frac{\phi(f)}{f}\nabla^i\nabla^kf\nabla^jC_{jki}-\int_M\frac{\phi(f)}{f^2}\nabla^if\nabla^kf\nabla^jC_{jki}\\
&+&\int_M\frac{\phi(f)}{f}\nabla^kf\nabla^i\nabla^jC_{jki}.
\end{eqnarray*}
Hence, from \eqref{zero} and the symmetries of the Cotton tensor by integration, we have
\begin{eqnarray}\label{enxu}
 \frac{1}{2}\int_M|C|^2\phi(f)+\int_M\frac{\phi(f)}{f^2}\nabla^if\nabla^kf\nabla^jC_{jki}&=&\int_M\frac{\phi(f)}{f}\nabla^kf\nabla^i\nabla^jC_{jki}+4\int_M\frac{\phi(f)}{f^3}\nabla^jf\nabla^i\psi\nabla^k\psi C_{jki}\nonumber\\
&-&2\int_M\frac{\phi(f)}{f^2}\nabla^j\nabla^i\psi\nabla^k\psi C_{jki}
=\int_M\frac{\phi(f)}{f}\nabla^kf\nabla^i\nabla^jC_{jki}\\
 &+&2\int_M\frac{\phi(f)}{f^3}(f\nabla^j\psi\nabla^k\nabla^i\psi+2\nabla^jf\nabla^k\psi\nabla^i\psi)C_{jki}\nonumber.
\end{eqnarray}

Now, considering that $\psi=\psi(f)$, we deduce
\begin{equation*}
    \begin{aligned}
        \int_M\frac{\phi(f)}{f^3}(f\nabla^j\psi&\nabla^k\nabla^i\psi+2\nabla^jf\nabla^k\psi\nabla^i\psi)C_{jki}\\
        &=\int_M\frac{\phi(f)}{f^3}[f\dot{\psi}\nabla^jf(\dot{\psi}\nabla^k\nabla^if+\ddot{\psi}\nabla^if\nabla^kf)+2\dot{\psi}^2\nabla^jf\nabla^kf\nabla^if]C_{jki}\\
        &=\int_M\frac{\phi(f)}{f^2}\dot{\psi}^2\nabla^jf\nabla^k\nabla^ifC_{jki}.
    \end{aligned}
\end{equation*}
Again, from the symmetries of the Cotton tensor and renaming indices, we obtain
\begin{eqnarray*}
   \int_M\frac{\phi(f)}{f^3}(f\nabla^j\psi\nabla^k\nabla^i\psi+2\nabla^jf\nabla^k\psi\nabla^i\psi)C_{jki}&=&\int_M\frac{\phi(f)}{f^2}\dot{\psi}^2\nabla^jf\nabla^k\nabla^ifC_{jki}\\
        &=&\int_M\frac{\phi(f)}{f^2}\dot{\psi}(f)^2\nabla^kf\nabla^if\nabla^jC_{jki}.
\end{eqnarray*}
Thus, replacing the above equation in \eqref{enxu}, we get
\begin{eqnarray}\label{quase}
\frac{1}{2}\int_M|C|^2\phi(f)+\int_M\frac{\phi(f)}{f^2}(1-2\dot{\psi}(f)^2)\nabla^kf\nabla^if\nabla^jC_{jki}=\int_M\frac{\phi(f)}{f}\nabla^kf\nabla^i\nabla^jC_{jki}.
\end{eqnarray}

From now on, we will analyse just one part of the above equation. Since the Cotton tensor is trace-free and skew-symmetric, another integration by parts gives us
\begin{eqnarray}\label{eqr01}
\int_M\frac{\phi(f)}{f^2}(1-2\dot{\psi}(f)^2)\nabla^kf\nabla^if\nabla^jC_{jki}&=&-\int_M\left(\frac{\dot{\phi(f)}}{f^2}-\frac{2\phi(f)}{f^3}\right)(1-2\dot{\psi}(f)^2)\nabla^kf\nabla^jf\nabla^ifC_{jki}\nonumber\\
&+&4\int_M\frac{\phi(f)}{f^2}\dot{\psi}(f)\ddot{\psi}(f)\nabla^kf\nabla^jf\nabla^ifC_{jki}\nonumber\\
&-&\int_M\frac{\phi(f)}{f^2}(1-2\dot{\psi}(f)^2)\nabla^j\nabla^kf\nabla^ifC_{jki}\nonumber\\
&-&\int_M\frac{\phi(f)}{f^2}(1-2\dot{\psi}(f)^2)\nabla^kf\nabla^j\nabla^ifC_{jki}\nonumber\\
&=&\int_M\frac{\phi(f)}{f}(1-2\dot{\psi}(f)^2)R^{ji}\nabla^kfC_{kji}.\nonumber
\end{eqnarray}
In the last equality we used \eqref{1} and renamed indices. Now, since $M^n$ has zero radial Weyl curvature and the Cotton tensor is totally trace-free, from \eqref{ttt} and \eqref{fundamental}, we infer
\begin{eqnarray*}
R^{ji}\nabla^kfC_{kji}&=&\frac{1}{2}C_{kji}(R^{ji}\nabla^kf-R^{ki}\nabla^jf)\\
&=&-\frac{1}{2Q}C_{kji}V^{kji}\\
&=&-\frac{1}{2Q}f|C|^2,
\end{eqnarray*}
where $Q$ is the same as given in Lemma \ref{primeiro1}, i.e., $Q=\frac{n-1}{n-2}-2\dot{\psi}(f)^2$. Therefore, we have
\begin{eqnarray*}
\int_M\frac{\phi(f)}{f^2}(1-2\dot{\psi}(f)^2)\nabla^kf\nabla^if\nabla^jC_{jki}&=&-\frac{1}{2}\int_M\frac{\phi(f)}{Q}(1-2\dot{\psi}(f)^2)|C|^2\\
&=&-\frac{1}{2}\int_M|C|^2\phi(f)\left[\frac{(n-2)(1-2\dot{\psi}(f)^2)}{n-1-2(n-2)\dot{\psi}(f)^2}\right].
\end{eqnarray*}
Now, from \eqref{2psiponto} we can conclude that
\begin{eqnarray*}
\int_M\frac{\phi(f)}{f^2}(1-2\dot{\psi}(f)^2)\nabla^kf\nabla^if\nabla^jC_{jki}&=&-\frac{n-2}{2(n-1)^2\sigma}\int_M\phi(f)\left[f^2+(n-1)\sigma\right] |C|^2.
\end{eqnarray*}

Replacing it in \eqref{quase}, we obtain
\begin{eqnarray*}
\frac{1}{2(n-1)^2\sigma}\int_M|C|^2\phi(f)\left[(n-1)\sigma-(n-2)f^2\right]=\int_M\frac{\phi(f)}{f}\nabla^kf\nabla^i\nabla^jC_{jki}.
\end{eqnarray*}
Using \eqref{cw} the result holds. 
\end{proof}

Next, we will take an appropriate $\phi(f)$ satisfying the conditions of integrability in the Theorem \ref{teo3} (cf. Remark \ref{remarktheo}).

\begin{theorem}\label{teoharmonico}
Let $(M^n,\,g,\,f,\,\psi)$, $n \geq 4$, be a closed electrovacuum space satisfying \eqref{s1}, \eqref{ruim} and \eqref{radial} with fourth-order divergence-free Weyl tensor, i.e., $\textnormal{div}^4W = 0$. Then, the Weyl tensor is harmonic, i.e., $\textnormal{div}W=0$.
\end{theorem}

\begin{proof}
Since $M$ is compact without boundary, we assume $\phi(f)=f^4$, then from Theorem \ref{teo3}, we obtain
\begin{eqnarray*}
\frac{1}{2(n-1)^2\sigma}\int_M|C|^2\phi(f)\left[(n-1)\sigma-(n-2)f^2\right]&=&-\frac{n-2}{n-3}\int_Mf^3\nabla^kf\nabla^i\nabla^j\nabla^lW_{jkil}\\
&=&-\frac{n-2}{4(n-3)}\int_M\nabla^if^4\nabla^k\nabla^j\nabla^lW_{jkil}\\
&=&\frac{n-2}{4(n-3)}\int_Mf^4\nabla^i\nabla^k\nabla^j\nabla^lW_{jkil},
\end{eqnarray*}
By hypothesis $\textnormal{div}^4W = 0$, then the right-hand side in the last equation is identically zero, i.e.,
\begin{eqnarray*}
\int_M|C|^2f^4\left[(n-1)\sigma-(n-2)f^2\right]&=&0.
\end{eqnarray*}

From Theorem \ref{psi de f}, since $M$ is subextremal (i.e., $\sigma<0$), we have
\begin{eqnarray*}
0\leq\int_M|C|^2f^4\left[(n-2)f^2-(n-1)\sigma\right]&=&0.
\end{eqnarray*}
that is, the Cotton tensor is identically zero. Therefore, from \eqref{cw} the result holds.
\end{proof}

\begin{theorem}\label{proper}
Let $(M^n,\,g,\,f,\,\psi)$, $n \geq 4$, be an electrovacuum space satisfying \eqref{s1}, \eqref{ruim} and \eqref{radial} with fourth-order divergence-free Weyl tensor, i.e., $\textnormal{div}^4W = 0$. If $f$ is a proper function, then the Weyl tensor is harmonic, i.e., $\textnormal{div}W=0$.
\end{theorem}
\begin{proof}
Let $s>0$ be a real number fixed and we take $\chi\in C^3$ a real non-negative function defined by $\chi=1$ in {$[0,s]$}, $\chi'\leq0$ in $[s,2s]$ and $\chi=0$ in $[2s,+\infty]$. Since $f$ is a proper function, we have that $\phi(f)=f^4\chi(f)$ has compact support in $M$ for $s>0$. From Theorem \ref{teo3}, we get
\begin{eqnarray*}
\frac{1}{2(n-1)^2\sigma}\int_M|C|^2f^4\chi(f)\left[(n-1)\sigma-(n-2)f^2\right]&=&-\frac{n-2}{n-3}\int_Mf^3\chi(f)\nabla^kf\nabla^i\nabla^j\nabla^lW_{jkil}\\&=&-\frac{n-2}{4(n-3)}\int_M\chi(f)\nabla^if^4\nabla^k\nabla^j\nabla^lW_{jkil}\\
&=&\frac{n-2}{4(n-3)}\int_M\chi(f)f^4\nabla^i\nabla^k\nabla^j\nabla^lW_{jkil}\\
&+&\frac{n-2}{4(n-3)}\int_M\dot{\chi}(f)f^4\nabla^if\nabla^k\nabla^j\nabla^lW_{jkil}.
\end{eqnarray*}

In the last equality we use integration by parts. Now, we take $\phi(f)=f^5\dot{\chi}(f)$ in the Theorem \ref{teo3} and since   $\textnormal{div}^4W=0$, we obtain
\begin{eqnarray*}
\frac{1}{2(n-1)^2\sigma}\int_M|C|^2f^4\chi(f)\left[(n-1)\sigma-(n-2)f^2\right]&=&-\frac{1}{8(n-1)^2\sigma}\int_M|C|^2f^5\dot{\chi(f)}\left[(n-1)\sigma-(n-2)f^2\right].
\end{eqnarray*}
Hence, 
\begin{equation*}
    \begin{aligned}
       \int_Mf^4|C|^2[\chi(f)+\frac{1}{4}f\dot{\chi}(f)]\left[(n-1)\sigma-(n-2)f^2\right]=0.
    \end{aligned}
\end{equation*}

Define $M_s=\{x\in M; f(x)\leq s\}$. We have, by definition, $\chi(f)+\frac{1}{4}f\dot{\chi}(f)=1$ on the compact set $M_s$. Thus, on $M_s$, since $M$ is subextremal,
$$0\leq\int_{M_s}f^4|C|^2\left[(n-2)f^2-(n-1)\sigma\right]=0,$$
i.e.,  $C=0$ in $M_s$. Taking $s\rightarrow+\infty$, we obtain that $C=0$ on $M$. 
\end{proof}
%Finally, Theorem \ref{fiber007-2} follows combining this result with Theorem \ref{fiber007-1}.

Now, we are ready to present the proof of Corollary \ref{fiber007-2} that will be stated again here for the sake of the reader’s convenience.

\begin{corollary}[Corollary \ref{fiber007-2}]\label{fiber0072}
Let $(M^n, g, f, \psi)$, $n>3$, be a complete subextremal electrovacuum space with fourth-order divergence free Weyl curvature and zero radial Weyl curvature such that the electric potential $\psi$ is in the Reissner-Nordstr\"om class (i.e., satisfying Equation \eqref{ruim}). Around any regular point of $f$, if $f$ is a proper function, then the manifold is locally a warped product with $(n-1)$-dimensional Einstein fibers.
\end{corollary}

\begin{proof}
This result follows combining the Theorem \ref{fiber007-1} with the Theorem \ref{proper}.

\end{proof}

\subsection{Third-order divergence free Cotton tensor}\label{sectionprincipaldim3}

 In this subsection, we will return to our results and study them in dimension $n=3$. Firstly, it is important to point out that the lemmas \ref{segundo}, \ref{terceiro} and \ref{quarto} are not valid in dimension $n=3$ due equation \eqref{cw}, which was used in their demonstrations. However, we can prove another version of them in a convenient way. Other point is the fact that the Theorem \ref{teo3} is not valid in dimension $n=3$, but the main issue here is that the Weyl tensor vanishes in dimension three. Nonetheless, the computations are very much similar to the previous results proved in dimension more than $3$. We will prove all those results for $n=3$ for the sake of completeness of the text.

After these considerations, we can proceed with our results. To that end, since the Weyl tensor vanishes identically in dimension $n=3$, we can observe that equation \eqref{ttt} becomes
 \begin{equation}\label{ttt3}
     fC_{ijk}=V_{ijk}.
 \end{equation}
Consequently, we have the following lemma.
\begin{lemma}\label{lemma3.3}
Let $(M^3,\,g,\,f,\,\psi)$ be an electrovacuum space. Then, 
\begin{equation*}
    C_{kji}R^{ik}=\nabla^i\nabla^k\left(\frac{V_{kij}}{f}\right).
\end{equation*}
\end{lemma}
\begin{proof}
In fact, from \eqref{bach3} and \eqref{ttt3}, we obtain
\begin{equation*}
    B_{ij}=\nabla^kC_{kij}=\nabla^k\left(\frac{V_{kij}}{f}\right).
\end{equation*}
Taking the derivative over $i$, we have
\begin{equation*}
    \nabla^iB_{ij}=\nabla^i\nabla^k\left(\frac{V_{kij}}{f}\right).
\end{equation*}
Since $n=3$, from \eqref{bc}, using \eqref{soma} and \eqref{anula} after renamed the indices, we infer
\begin{equation*}
    \nabla^iB_{ij}=-C_{jik}R^{ik}=-C_{jki}R^{ik}=C_{kji}R^{ik}.
\end{equation*}

Thus, combing these two last relations the result holds.
\end{proof}

\begin{lemma}\label{lemma4.3}
Let $(M^3,\,g,\,f,\,\psi)$ be an electrovacuum space. Then,
\begin{equation*}
    \begin{aligned}
        \frac{1}{2}|C|^2+R^{ik}\nabla^jC_{jki}=-\nabla^j\nabla^i\nabla^k\left(\frac{V_{kij}}{f}\right).
    \end{aligned}
\end{equation*}
\end{lemma}
\begin{proof}
Taking the divergence in Lemma \ref{lemma3.3}, we get
\begin{equation*}
    C_{kji}\nabla^jR^{ik}+R^{ik}\nabla^jC_{kji}=\nabla^j\nabla^i\nabla^k\left(\frac{V_{kij}}{f}\right).
\end{equation*}
Using \eqref{rc}, we have
\begin{equation*}
    \frac{1}{2}C_{kji}(\nabla^jR^{ik}-\nabla^kR^{ij})+R^{ik}\nabla^jC_{kji}=\nabla^j\nabla^i\nabla^k\left(\frac{V_{kij}}{f}\right).
\end{equation*}
Now, since the Cotton tensor is trace-free, from  \eqref{ct} and renaming the indices, we obtain
\begin{equation*}
    -\frac{1}{2}C_{kji}C^{kji}-R^{ik}\nabla^jC_{jki}=\nabla^j\nabla^i\nabla^k\left(\frac{V_{kij}}{f}\right).
\end{equation*}
Therefore, the result holds.
\end{proof}
\begin{theorem}\label{teo3.3'}
Let $(M^3,\,g,\,f,\,\psi)$ be an electrovacuum space satisfying \eqref{ruim}. For every $\phi:\mathbb{R}\rightarrow\mathbb{R}$, $C^2$
function with $\phi(f)$ having compact support $K\subseteq M.$ Then,

\begin{eqnarray*}
    \frac{1}{8\sigma}\int_M|C|^2\phi(f)[2\sigma-f^2]=\int_M\frac{\phi(f)}{f}\nabla^kf\nabla^i\nabla^jC_{jki}.
\end{eqnarray*}
where $\sigma$ is a non-null constant.
\end{theorem}
\begin{proof}
The idea is proceed as in Theorem \ref{teo3}. From Lemma \ref{lemma4.3}, we obtain
\begin{equation*}
    \begin{aligned}
        \frac{1}{2}|C|^2\phi(f)+\phi(f)R^{ik}\nabla^jC_{jki}=-\phi(f)\nabla^j\nabla^i\nabla^k\left(\frac{V_{kij}}{f}\right).
    \end{aligned}
\end{equation*}

Hence, upon integratity this expression, we get
\begin{equation*}
    \begin{aligned}
        \frac{1}{2}\int_M|C|^2\phi(f)+\int_M\phi(f)R^{ik}\nabla^jC_{jki}=\int_M\dot{\phi}(f)\nabla^jf\nabla^i\nabla^k\left(\frac{V_{kij}}{f}\right).
    \end{aligned}
\end{equation*}
Then, from Lemma \ref{lemma3.3} and the symmetries of $C_{ijk},$ we have 
\begin{equation*}
    \begin{aligned}
        \frac{1}{2}\int_M|C|^2\phi(f)+\int_M\phi(f)R^{ik}\nabla^jC_{jki}=&-\int_M\dot{\phi}(f)\nabla^jfC_{jki}R^{ik}.
    \end{aligned}
\end{equation*}

Now, from \eqref{1} and the fact that $C_{ijk}$ is trace-free and skew-symmetric we obtain the following identity 

\begin{equation*}
    \begin{aligned}
        \frac{1}{2}\int_M|C|^2\phi(f)+\int_M\frac{\phi(f)}{f}&(\nabla^i\nabla^kf-\frac{2}{f}\dot{\psi}(f)^2\nabla^if\nabla^kf)\nabla^jC_{jki}\\
        =&-\int_M\frac{\dot{\phi}(f)}{f}\nabla^jf(\nabla^k\nabla^if-\frac{2}{f}\dot{\psi}(f)^2\nabla^if\nabla^kf)C_{jki}\\
        =&-\int_M\frac{\dot{\phi}(f)}{f}\nabla^jf\nabla^k\nabla^ifC_{jki}\\
        =&\int_M\left(\frac{\ddot{\phi}(f)}{f}-\frac{\dot{\phi}(f)}{f^2}\right)\nabla^jf\nabla^kf\nabla^ifC_{jki}\\
        &+\int_M\frac{\dot{\phi}(f)}{f}\nabla^k\nabla^jf\nabla^ifC_{jki}+\int_M\frac{\dot{\phi}(f)}{f}\nabla^jf\nabla^if\nabla^kC_{jki}\\
        =&\int_M\frac{\dot{\phi}(f)}{f}\nabla^jf\nabla^if\nabla^kC_{jki}.
    \end{aligned}
\end{equation*}

Note that in the last equality we used \eqref{hessiana}. From now, we rename the indices and integrating by parts again, we infer

\begin{equation*}
    \begin{aligned}
        \frac{1}{2}\int_M|C|^2\phi(f)&+\int_M\frac{\phi(f)}{f}\nabla^i\nabla^kf\nabla^jC_{jki}-2\int_M\frac{\phi(f)}{f^2}\dot{\psi}(f)^2\nabla^if\nabla^kf\nabla^jC_{jki}\\
        =&-\int_M\frac{\dot{\phi}(f)}{f}\nabla^kf\nabla^if\nabla^jC_{jki}=-\int_M\frac{\nabla^i\phi(f)}{f}\nabla^kf\nabla^jC_{jki}\\
        =&\int_M\frac{\phi(f)}{f}\nabla^k\nabla^if\nabla^jC_{jki}-\int_M\frac{\phi(f)}{f^2}\nabla^kf\nabla^if\nabla^jC_{jki}\\
        &+\int_M\frac{\phi(f)}{f}\nabla^kf\nabla^i\nabla^jC_{jki}.
    \end{aligned}
\end{equation*}
Thus, 
\begin{eqnarray}\label{integral}
    \frac{1}{2}\int_M|C|^2\phi(f)+\int_M\frac{\phi(f)}{f^2}(1-2\dot{\psi}(f)^2)\nabla^kf\nabla^if\nabla^jC_{jki}=\int_M\frac{\phi(f)}{f}\nabla^kf\nabla^i\nabla^jC_{jki}.
\end{eqnarray}
Furthermore, from the proof of Theorem \ref{teo3} (Equation \ref{eqr01}), we get
\begin{eqnarray*}
\int_M\frac{\phi(f)}{f^2}(1-2\dot{\psi}(f)^2)\nabla^kf\nabla^if\nabla^jC_{jki}&=&\int_M\frac{\phi(f)}{f}(1-2\dot{\psi}(f)^2)R^{ji}\nabla^kfC_{kji}.
\end{eqnarray*}
Again, as we did in Theorem \ref{teo3}, from \eqref{fundamental} and \eqref{ttt3}, we have
\begin{eqnarray*}
R^{ji}\nabla^kfC_{kji}=-\frac{1}{2Q}f|C|^2.
\end{eqnarray*}
Note that in dimension three, from Lemma \ref{primeiro1} and \eqref{2psiponto}, we obtain, respectively,
\begin{eqnarray*}
Q=2(1-\dot{\psi}(f)^2)
\end{eqnarray*}
and
\begin{eqnarray*}
\dot{\psi}(f)^2=\frac{f^2}{f^2-2\sigma};\quad \mbox{where} \quad \sigma\neq0.
\end{eqnarray*}
Finally,
\begin{eqnarray*}
\int_M\frac{\phi(f)}{f^2}(1-2\dot{\psi}(f)^2)\nabla^kf\nabla^if\nabla^jC_{jki}&=&-\frac{1}{8\sigma}\int_M|C|^2\phi(f)\left[f^2+2\sigma\right] .
\end{eqnarray*}
Therefore, replacing the above equation in \eqref{integral} the result holds.
\end{proof}

\begin{theorem}\label{teoharmonico'}
Let $(M^3,\,g,\,f,\,\psi)$ be a closed subextremal electrovacuum space satisfying \eqref{ruim} with third-order divergence-free Cotton tensor, i.e., $\textnormal{div}^3C = 0$. Then, the Cotton tensor is identically zero, i.e., $(M^3,\,g)$ is locally conformally flat.
\end{theorem}

\begin{proof}
Since $M$ is compact without boundary, we assume $\phi(f)=f^4$, then from Theorem \ref{teo3}, we obtain
\begin{eqnarray*}
\frac{1}{8\sigma}\int_M|C|^2\phi(f)\left[2\sigma-f^2\right]&=&\int_Mf^3\nabla^kf\nabla^i\nabla^jC_{jki}\\
&=&\frac{1}{4}\int_M\nabla^if^4\nabla^k\nabla^jC_{jki}\\
&=&-\frac{1}{4}\int_Mf^4\nabla^i\nabla^k\nabla^jC_{jki},
\end{eqnarray*}
Since  $\textnormal{div}^3C = 0$, the right-hand side is identically zero, i.e., 
\begin{eqnarray*}
\int_M|C|^2f^4\left[2\sigma-f^2\right]=0.
\end{eqnarray*}

Thus, considering $M$ is subextremal, the Cotton tensor must be zero. 
\end{proof}

\begin{theorem}\label{proper'}
Let $(M^3,\,g,\,f,\,\psi)$ be a subextremal electrovacuum space satisfying \eqref{ruim} with third-order divergence-free Cotton tensor, i.e., $\textnormal{div}^3C = 0$. If $f$ is a proper function, then the Cotton tensor is identically zero, i.e., $(M^3,\,g)$ is locally conformally flat.
\end{theorem}
\begin{proof}
Let $s>0$ be a real number fixed and we take $\chi\in C^3$ a real non-negative function defined by $\chi=1$ in {$[0,s]$}, $\chi'\leq0$ in $[s,2s]$ and $\chi=0$ in $[2s,+\infty]$. Since $f$ is a proper function, we have that $\phi(f)=f^4\chi(f)$ has compact support in $M$ for $s>0$. From Theorem \ref{teo3}, we get
\begin{eqnarray*}
\frac{1}{8\sigma}\int_M|C|^2f^4\chi(f)\left[2\sigma-f^2\right]&=&\int_Mf^3\chi(f)\nabla^kf\nabla^i\nabla^jC_{jki}\\&=&\frac{1}{4}\int_M\chi(f)\nabla^if^4\nabla^k\nabla^jC_{jki}\\
&=&-\frac{1}{4}\int_M\chi(f)f^4\nabla^i\nabla^k\nabla^jC_{jki}\\
&+&\frac{1}{4}\int_M\dot{\chi}(f)f^4\nabla^if\nabla^k\nabla^jC_{jki}.
\end{eqnarray*}

In the last equality we used integration by parts. Now, since $\textnormal{div}^3 C=0$ we take $\phi(f)=f^5\dot{\chi}(f)$ in the Theorem \ref{teo3} one more time to obtain 
\begin{eqnarray*}
\frac{1}{8\sigma}\int_M|C|^2f^4\chi(f)\left[2\sigma-f^2\right]&=&-\frac{1}{32\sigma}\int_M|C|^2f^5\dot{\chi(f)}\left[4\sigma-f^2\right],
\end{eqnarray*}
i.e.,
\begin{equation*}
    \begin{aligned}
       \int_Mf^4|C|^2[\chi(f)+\frac{1}{4}f\dot{\chi}(f)]\left[2\sigma-f^2\right]=0.
    \end{aligned}
\end{equation*}

Let be $M_s$ defined as in Theorem \ref{proper}, i.e., $M_s=\{x\in M; f(x)\leq s\}$. We have, by definition, $\chi(f)+\frac{1}{4}f\dot{\chi}(f)=1$ on the compact set $M_s$. Thus, on $M_s$, since $M$ is subextremal ($\sigma<0$),
$$0\leq\int_{M_s}f^4|C|^2\left[f^2-2\sigma\right]=0.$$
Therefore, $C=0$ in $M_s$. Taking $s\rightarrow+\infty$, we obtain that $C=0$ on $M$.
\end{proof}

Finally, we prove Corollary \ref{fiber007-3}, whose statement is as follows.

\begin{corollary}[Corollary \ref{fiber007-3}]

Let $(M^3, g, f, \psi)$ be a complete subextremal electrovacuum space with third-order divergence free Cotton tensor such that $\psi$ is in the Reissner-Nordstr\"om class. Around any regular point of $f$, if $f$ is a proper function, then the manifold is locally an Einstein manifold, i.e., $(M^3,\,g)$ is locally isometric to either $\mathbb{R}^{3}$ or $\mathbb{S}^{3}$.
\end{corollary}

\begin{proof}
%[\bf Proof of Corollary \ref{fiber007-3}]
This result is a consequence of Theorem \ref{fiber007-1}  and Theorem \ref{proper'}.
\end{proof}

\bibliographystyle{unsrt}  
%\bibliography{references}  %%% Remove comment to use the external .bib file (using bibtex).
%%% and comment out the ``thebibliography'' section.

%%% Comment out this section when you \bibliography{references} is enabled.

\end{document}